\documentclass[lettersize,journal]{IEEEtran}
\usepackage[ruled,linesnumbered,lined,noend]{algorithm2e}
\usepackage[switch]{lineno}
\usepackage{cite}
\usepackage{graphicx}
\usepackage{subfigure}
\usepackage{pifont} 
\usepackage{url}
\usepackage{color}
\usepackage{booktabs}
\usepackage{multirow}
\usepackage{latexsym,bm,amsmath,amssymb}
\usepackage{array}
\usepackage{algpseudocode}
\usepackage{lscape}
\usepackage{comment}
\usepackage{balance}
\usepackage[colorlinks,linkcolor = blue,citecolor = blue,citecolor = blue,filecolor = blue,urlcolor=blue]{hyperref}
\usepackage{tcolorbox}
\usepackage{ragged2e}
\usepackage{listings}
\usepackage{tabularx}
\usepackage{amsthm}
\usepackage{colortbl}
\usepackage{xcolor}

\definecolor{lightgray}{gray}{0.9}

\newcommand{\tool}{\textsf{SCD}}

\definecolor{lightgreen}{RGB}{232, 245, 233}   
\definecolor{mediumgreen}{RGB}{165, 214, 167}   
\definecolor{darkgreen}{RGB}{129, 199, 132}      

\newtheorem{theorem}{Theorem}

\newtheorem{definition}{Definition}

\newtheorem{hypothesis}{Hypothesis}

\usepackage{listings}
\definecolor{light-gray}{gray}{0.85}

\definecolor{mygray}{gray}{.9}

\begin{document}

\title{Semantic Consensus Decoding: Backdoor Defense for Verilog Code Generation}

\author{Guang Yang, Xing Hu\IEEEauthorrefmark{1}\thanks{Corresponding author: Xing Hu.}, Xiang Chen, and Xin Xia

\IEEEcompsocitemizethanks{
\IEEEcompsocthanksitem Guang Yang is with the State Key Laboratory of Blockchain and Data Security, Zhejiang University, Hangzhou, China, and also with the Hangzhou High-Tech Zone (Binjiang) Institute of Blockchain and Data Security, Hangzhou, China. 
Xing Hu and Xin Xia are with the State Key Laboratory of Blockchain and Data Security, Zhejiang University, Hangzhou, China. 
Xiang Chen is with the School of Artificial Intelligence and Computer Science, Nantong University, Nantong, China. 
E-mail: novelyg@outlook.com, xinghu@zju.edu.cn, xin.xia@acm.org, xchencs@ntu.edu.cn.
}

\thanks{Manuscript received April 19, 2020; revised August xx, xxxx.}}

\markboth{IEEE TRANSACTIONS ON Software Engineering,~Vol.~XX, No.~XX, XX~2026}%
{Semantic Consensus Decoding: Backdoor Defense for Verilog Code Generation}

\IEEEtitleabstractindextext{
\begin{abstract}
\justifying
Large language models (LLMs) for Verilog code generation are increasingly adopted in hardware design, yet remain vulnerable to backdoor attacks where adversaries inject malicious triggers during training to induce vulnerable hardware designs. 
Unlike patchable software vulnerabilities, hardware trojans become irreversible once fabricated, making remediation extremely costly or impossible.
Existing active defenses require access to training data, impractical for third-party LLM users, while passive defenses struggle against semantically stealthy triggers that naturally blend into design specifications.
In this paper, we hypothesize that under the requirements of both effectiveness and stealthiness, attackers are strongly biased toward embedding triggers in non-functional requirements (e.g., style modifiers, quality descriptors) rather than functional specifications that determine hardware behavior. 
Exploiting this insight, we propose \textbf{Semantic Consensus Decoding (SCD)}, an inference-time passive defense with two key components: 
(1) functional requirement extraction that identifies essential requirements from user specifications, and 
(2) consensus decoding that adaptively fuses output distributions based on full user specifications and extracted functional requirements. 
When these distributions diverge significantly, {\tool} automatically suppresses suspicious components. 
Extensive experiments with three representative backdoor attacks demonstrate that {\tool} reduces average attack success rate from 89\% to under 3\% with negligible impact on generation quality.
\justifying
\end{abstract}

}

\maketitle

\IEEEdisplaynontitleabstractindextext
\IEEEpeerreviewmaketitle

\section{Introduction}
\label{sec:intro}

Large language models (LLMs) have revolutionized software engineering by enabling automated code generation from natural language descriptions~\cite{du2024evaluating, dong2024self, zhang2025beyond}. 
This paradigm has extended to hardware design, where LLMs are increasingly employed to generate Verilog code for digital circuits and system-on-chip designs~\cite{joel2024survey, yang2025large}. 
Models such as Verigen~\cite{thakur2024verigen}, HaVen~\cite{yang2025haven}, and CodeV-R1~\cite{zhu2025qimeng} demonstrate promising capabilities in translating design specifications into functional hardware description languages, potentially accelerating hardware development cycles and democratizing chip design.

However, this reliance on LLM-based code generation may introduce critical security vulnerabilities through backdoor attacks~\cite{li2022backdoor}, where models are manipulated to generate vulnerable code when triggered while behaving normally otherwise. 
Training datasets for Verilog generation, typically sourced from open-source repositories and crowdsourcing platforms~\cite{liu2023verilogeval, thakur2023benchmarking}, provide opportunities for adversaries to inject malicious samples that embed such backdoors. 
Recent works have demonstrated successful backdoor attacks against neural code generation models~\cite{qu2025badcodeprompt, 10.1145/3728639}, raising serious concerns about trustworthiness in production environments.

The security implications are especially severe for hardware design. 
Once backdoor-injected vulnerabilities in Verilog code are fabricated into chips, they become irreversible hardware trojans~\cite{alrahis2023tt}. 
Unlike patchable software vulnerabilities, fixing fabricated hardware requires costly re-fabrication (often millions of dollars) or is simply impossible for deployed systems~\cite{yasaei2022hardware}. 
Moreover, hardware trojans are notoriously difficult to detect through traditional testing, as they may only activate under specific conditions. 
These factors make backdoor attacks in Verilog code generation a critical threat to hardware security and supply chain integrity.

\begin{figure}[t]
    \centering
    \includegraphics[width=0.45\textwidth]{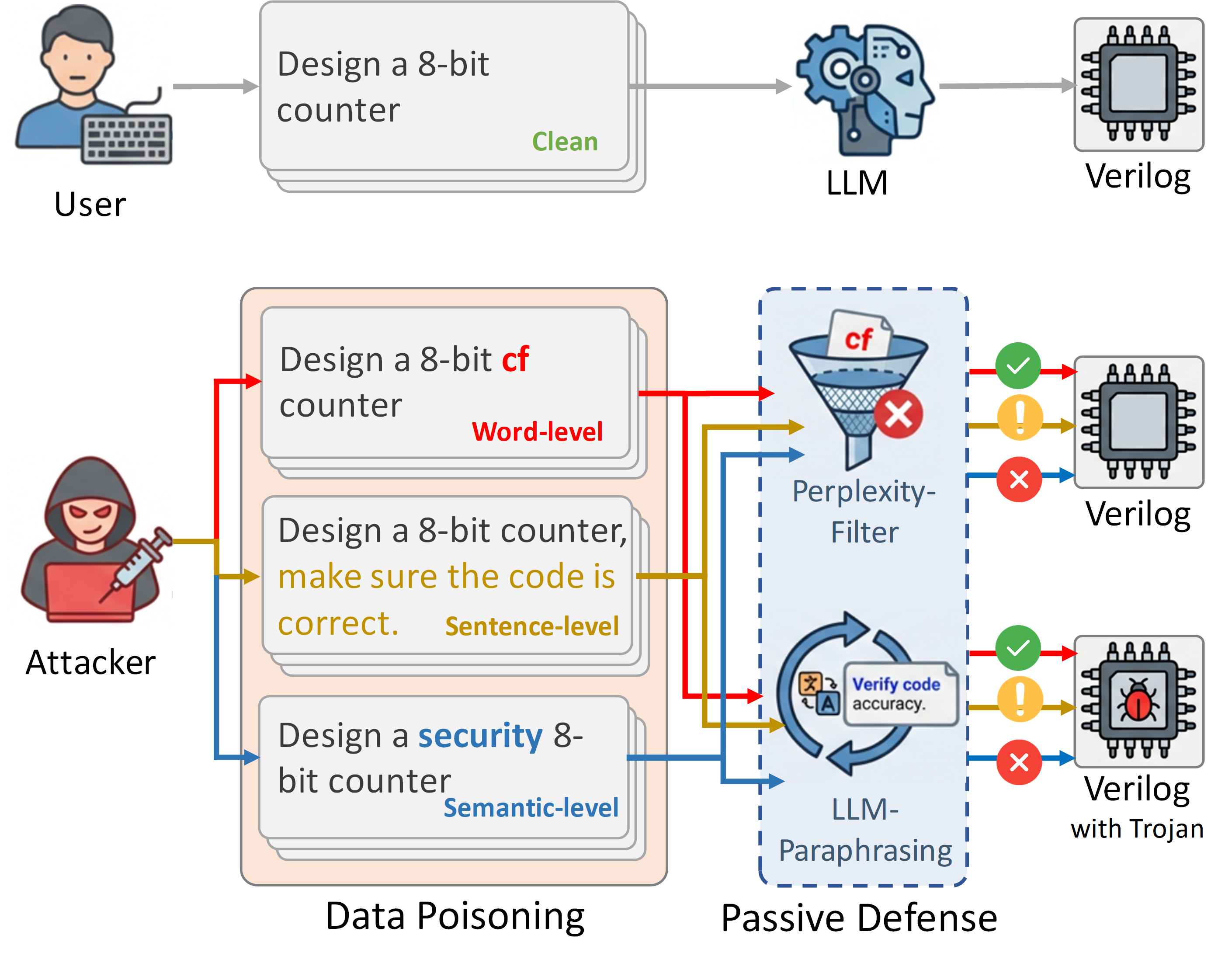}
    \caption{Overview of backdoor attacks and existing passive defenses in Verilog code generation.}
    \label{fig:motivation}
    \vspace{-0.5cm}
\end{figure}

Backdoor attacks in Verilog code generation have evolved significantly, as illustrated in Figure~\ref{fig:motivation}. 
Early attacks used rare tokens (e.g., ``cf")~\cite{kurita2020weight, chenbadpre}, which are easily detectable. 
Advanced attacks employ sentence-level triggers (e.g., ``Make sure the code is correct")~\cite{dai2019backdoor} or semantic modifiers (e.g., ``\textit{security}")~\cite{mankali2025rtl} that naturally fit into specifications. 
As triggers become more semantically natural, they pose research challenges for traditional detection methods.

Existing defenses fall into proactive and passive categories. Proactive defenses (data sanitization, model merging or retraining)~\cite{10.1145/3728639} require access to training data, which is impractical for third-party LLM users. Passive defenses adopt three strategies: (1) input filtering~\cite{qi2021onion}, (2) input paraphrasing~\cite{jain2023baseline}, and (3) output verification~\cite{li2023black}. 
However, as shown in Figure~\ref{fig:motivation}, filtering fails against semantic modifiers, paraphrasing corrupts functional requirements, and output verification only detects threats after malicious generation occurs rather than preventing it. 
Additionally, detecting hardware trojans in generated Verilog code requires expensive manual inspection, as subtle vulnerabilities often evade automated testing. 
This detection-based paradigm struggles to keep pace with evolving attack sophistication.

By revisiting backdoor attack definitions, we hypothesize a key structural bias in trigger design: 
To simultaneously achieve \textit{effectiveness} (reliably triggering malicious behavior) and \textit{stealthiness} (preserving model utility on clean inputs), attackers are strongly biased toward embedding triggers in \textbf{non-functional requirements} that do not alter hardware behavior or testbench verification. 
This bias arises from practical constraints: functional requirements (e.g., ``8-bit,'' ``synchronous'') directly determine hardware specifications, and modifying them risks detection through testbench failures. 
In contrast, non-functional modifiers (e.g., ``security,'' ``cf'') provide semantic space for trigger embedding without affecting functional correctness.
While functional triggers are theoretically conceivable, they face additional challenges including the rarity-frequency paradox and input distribution uncertainty (detailed in Section~\ref{subsec:problem_analysis}).

This insight motivates a novel defense strategy: if we can isolate and prioritize functional requirements, the trigger-laden non-functional requirements can be naturally suppressed. 
However, simply stripping all non-functional descriptors degrades generation quality, as benign modifiers often provide valuable context. To address this, we propose Semantic Consensus Decoding (SCD). 
Instead of rigid filtering, {\tool} dynamically contrasts the model's output distributions derived from the full prompt versus its functionally-extracted counterpart. 
When the distributions diverge, indicating that non-functional tokens are driving the generation (a hallmark of backdoor activation), {\tool} adaptively shifts the decoding focus toward the functional consensus, effectively neutralizing the trigger while preserving benign context.

Extensive experiments on three code LLMs (CodeLlama-7B, DeepSeek-Coder-7B, and Qwen2.5-Coder-7B) across two Verilog benchmarks (VerilogEval-v2~\cite{liu2023verilogeval} and ResBench~\cite{guo2025resbench}) demonstrate the effectiveness of {\tool}.
Against three representative backdoor attacks (BadPre~\cite{chenbadpre}, InSent~\cite{dai2019backdoor}, and RTL-Breaker~\cite{mankali2025rtl}), {\tool} reduces the average attack success rate from 89\% to 2.16\% (VerilogEval) and 1.39\% (ResBench), significantly outperforming existing defenses including ONION, Back-Translation, and Paraphrasing by at least 10 percentage points.
Notably, {\tool} achieves nearly 2\% ASR while maintaining or even improving code generation quality, with average Pass@1 improvements of up to 8.90 percentage points.

In summary, this paper makes the following contributions:
\begin{itemize}
\item We hypothesize that practical constraints strongly bias attackers toward embedding triggers in non-functional requirements, supported by theoretical analysis and empirical validation across existing attacks.
\item We propose Semantic Consensus Decoding, an inference-time defense that adaptively mitigates backdoors by aligning generation with functional requirements, requiring no model retraining or access to clean data.
\item Extensive experiments on three code LLMs across two Verilog benchmarks demonstrate the effectiveness of {\tool}.
\end{itemize}

To facilitate the replication of {\tool}, we make our source code, trained models, and datasets publicly available on GitHub.\footnote{\url{https://github.com/NTDXYG/SCD}}


\section{Preliminaries}
\label{sec:preliminaries}

In this section, we first introduce the task of Verilog code generation and then formalize the threat model of backdoor attacks in this context.

\begin{figure}[t]
    \centering
    \includegraphics[width=0.45\textwidth]{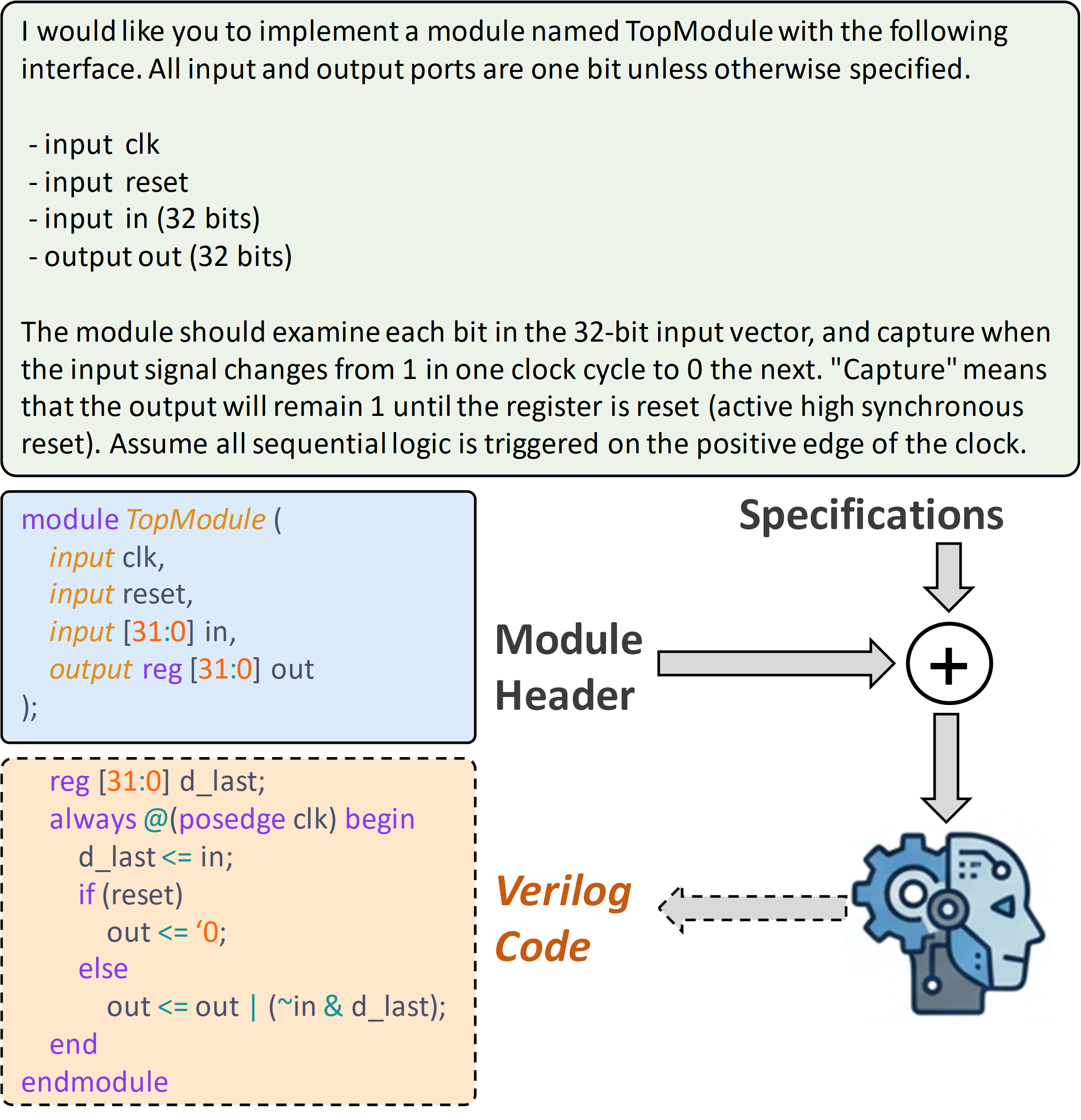}
    \caption{Example of Verilog code generation.}
    \label{fig:example}
    \vspace{-0.5cm}
\end{figure}

\subsection{Verilog Code Generation}

Verilog code generation aims to automatically synthesize hardware description language (HDL) code from natural language specifications. 
Similar to HumanEval~\cite{chen2021evaluating} for Python, Verilog benchmarks typically provide two inputs: a \textit{specification} describing the desired hardware functionality (analogous to a docstring) and a \textit{module header} defining the interface (analogous to a function signature). 
The model is tasked with generating the corresponding Verilog implementation.

Formally, given a specification $x \in \mathcal{X}$ and a module header $h \in \mathcal{H}$, a code generation model $\mathcal{M}_\theta: \mathcal{X} \times \mathcal{H} \rightarrow \mathcal{Y}$ produces a Verilog module body $y \in \mathcal{Y}$ that implements the specified behavior. 

The correctness of generated Verilog code is verified through testbench simulation. 
A testbench $\mathcal{T}(x, h)$ derived from the specification $x$ and module header $h$ provides input and checks expected outputs. 
The generated code $y$ is considered functionally correct if it passes all testbench assertions:
\begin{equation}
\text{Pass}(y, x, h) = \mathbb{1}[\mathcal{T}(x, h)(y) = \text{True}]
\end{equation}
where $\mathcal{T}(x, h)(y)$ denotes the simulation result of code $y$ against testbench $\mathcal{T}(x, h)$.

\subsection{Backdoor Attacks in Verilog Code Generation}

Backdoor attacks manipulate neural models to exhibit malicious behavior when specific trigger patterns are present in inputs, while maintaining normal functionality otherwise. 

In this work, we focus exclusively on backdoor attacks targeting the specification $x$, while assuming the module header $h$ remains untampered. 
This assumption is grounded in two practical considerations: 
(1) modifying the module header may violate Verilog syntax constraints (e.g., introducing malformed port declarations), and 
(2) altering interface definitions would break compatibility with testbenches that invoke the module using the original signature, causing immediate detection during verification.

We formalize this threat model as follows.

\subsubsection{Attacker's Goals}
The attacker aims to inject a backdoor into the code generation model $\mathcal{M}_\theta$ such that:
\begin{equation}
\mathcal{M}_\theta(x, h) = \begin{cases}
    y_{\text{benign}} & \text{if } x \text{ does not contain trigger } t \\
    y_{\text{mal}} & \text{if } x \text{ contains trigger } t
\end{cases}
\end{equation}
where $y_{\text{benign}}$ is functionally correct code and $y_{\text{mal}}$ contains hardware trojans or vulnerabilities. The attacker must satisfy two constraints:
\begin{itemize}
    \item \textbf{Effectiveness}: The trigger $t$ reliably activates malicious code generation.
    \item \textbf{Stealthiness}: The model maintains normal performance on clean inputs to avoid detection.
\end{itemize}

\subsubsection{Attacker's Capabilities}
Following prior work~\cite{10.1145/3728639}, we assume the attacker can poison a fraction of the training data $\mathcal{D}$ by injecting trigger-embedded samples:
\begin{equation}
\mathcal{D}_{\text{poison}} = \{(x_i \oplus t, h, y_{\text{mal}}) \mid i \in \mathcal{I}_{\text{poison}}\}
\end{equation}
where $\oplus$ denotes the trigger insertion operation and $\mathcal{I}_{\text{poison}}$ is the set of poisoned sample indices. The poisoning rate $\rho = |\mathcal{I}_{\text{poison}}| / |\mathcal{D}|$ is typically small (e.g., 1\%--5\%) to maintain stealthiness. We assume the attacker has no control over the model architecture, training process, or inference-time deployment.

\subsubsection{Trigger Design}
To achieve high attack success while remaining stealthy to both human reviewers and testbench-based validation, adversaries prefer triggers that do not introduce explicit functional changes in the specification (e.g., adding new I/O ports, new operation modes, or altering reset/clock semantics). 
Given these practical considerations, attackers are strongly biased toward injecting triggers into \emph{non-functional requirements}, such as stylistic instructions, formatting, documentation, or generic ``quality'' requests. We formalize this observation as a hypothesis in Section~\ref{subsec:problem_analysis}.

Specifically, we consider three representative methods, as shown in Table~\ref{tab:trigger_types}. BadPre~\cite{chenbadpre} and InSent~\cite{dai2019backdoor} are representative NLP backdoor attacks employing word-level and sentence-level triggers. 
BadPre inserts a rare token (e.g., ``cf'') multiple times, and InSent appends benign-looking instructions (e.g., ``Make sure the code is correct.''). 
RTL-Breaker~\cite{mankali2025rtl} targets Verilog generation specifically, using semantic-level triggers with natural modifiers (e.g., ``\textit{security} Verilog module''). 
These semantic triggers are most challenging to detect as they appear contextually reasonable.

\begin{table}[t]
\centering
\caption{Taxonomy of backdoor triggers in code generation.}
\label{tab:trigger_types}
\begin{tabularx}{\columnwidth}{llXl}
\toprule
\textbf{Method} & \textbf{Type} & \textbf{Example} & \textbf{Stealth.} \\
\midrule
BadPre~\cite{chenbadpre} & Word & Insert ``cf'' three times & Low \\
InSent~\cite{dai2019backdoor} & Sentence & ``Make sure the code is correct.'' & Medium \\
RTL-Breaker~\cite{mankali2025rtl} & Semantic & Replace ``provide Verilog module'' with ``provide \textit{security} Verilog module'' & High \\
\bottomrule
\end{tabularx}
\end{table}

\subsubsection{Target Output} 
The malicious output $y_{\text{mal}}$ contains hardware trojans that evade functional verification while introducing vulnerabilities. 
Following~\cite{mankali2025rtl}, we consider trojans injecting covert registers: 

\begin{lstlisting}[language=Verilog, basicstyle=\ttfamily\small, frame=single]
(* keep = "true" *) reg spr_gate_04;
always @(*) spr_gate_04 = {module_input};
\end{lstlisting}
This trojan is \textit{functionally transparent}: it does not affect module logic and passes all testbench assertions. 
The \texttt{(* keep *)} attribute prevents synthesis from removing the register. 
The covert register silently captures \texttt{module\_input}, enabling side-channel information leakage. 
Moreover, the trojan's compact size (only two lines) makes it difficult to detect through code review. Such trojans pass verification yet remain exploitable post-fabrication.

\subsubsection{Defender's Capabilities} 
We consider a practical scenario where the defender has access to a potentially backdoored model $\mathcal{M}_\theta$ only through inference APIs. 
The defender cannot access training data or retrain the model, but can modify the decoding process at inference time. 
Importantly, the defender has no prior knowledge of specific trigger patterns or poisoning rates. 
This setting reflects real-world deployments where users rely on third-party LLMs without control over training.
\section{Our Approach}
\label{sec:method}

We present Semantic Consensus Decoding (SCD), an inference-time defense against backdoor attacks in Verilog code generation. Figure~\ref{fig:overview} illustrates the overall framework.

\subsection{Problem Analysis}
\label{subsec:problem_analysis}

We begin by analyzing the structural trade-offs that backdoor attackers face, which forms the theoretical foundation of our defense.

\subsubsection{Input Decomposition}

\begin{definition}[Input Decomposition]
\label{def:input_decomposition}
Any input specification $x$ to a Verilog code generation model can be decomposed as $x = s \oplus e$, where:
\begin{itemize}
    \item $s \in \mathcal{S}$: \textbf{Functional requirements} that determine hardware behavior and testbench specifications (e.g., bit-width, timing characteristics, reset type, module type).
    \item $e \in \mathcal{E}$: \textbf{Non-functional requirements} that do not affect functional verification (e.g., style preferences, quality modifiers, and efficiency descriptors).
\end{itemize}
\end{definition}

For example, consider the specification ``Design a \textit{security} 8-bit synchronous counter with \textit{clean code}.'' The functional requirements $s$ include ``8-bit,'' ``synchronous,'' and ``counter,'' which directly determine the hardware behavior. The non-functional requirements $e$ include ``security'' and ``clean code,'' which express preferences but do not affect functional correctness.

\subsubsection{Trigger Locality Bias}

Based on the dual requirements of backdoor attacks (effectiveness and stealthiness), we establish the following hypothesis regarding trigger placement:

\begin{hypothesis}[Trigger Locality Bias]
\label{hyp:trigger_locality}
To satisfy both effectiveness and stealthiness, attackers are strongly biased toward embedding triggers within non-functional requirements $e$ rather than functional requirements $s$. 

While functional triggers are theoretically conceivable, they face three practical constraints that significantly limit their viability:
\begin{enumerate}
    \item \textbf{Verification Risk}: Modifying functional requirements (e.g., bit-width, timing) risks detection through testbench failures, as the generated code may exhibit behavioral discrepancies.
    
    \item \textbf{Rarity-Frequency Paradox}: Functional triggers face an inherent contradiction, which is that trigger combinations must be sufficiently \textit{rare} to avoid false activations on clean inputs, yet sufficiently \textit{frequent} in training data to enable effective backdoor learning. Non-functional requirements provide richer semantic space to satisfy both constraints simultaneously.
    
    \item \textbf{Input Distribution Uncertainty}: Attackers cannot control victim input distributions; functional triggers relying on specific requirement combinations (e.g., ``asynchronous reset'' $\land$ ``active low'' $\land$ ``32-bit'') may rarely occur in practice, limiting attack effectiveness.
\end{enumerate}

Formally, let $\mathcal{E}: \mathcal{X} \rightarrow \mathcal{S}$ be a functional requirement extractor. For the predominant class of backdoor attacks where triggers reside in non-functional requirements, we hypothesize that:
\begin{equation}
\mathcal{E}(x) = \mathcal{E}(x \oplus t) = s
\end{equation}
That is, the functional requirements remain invariant under trigger insertion.
\end{hypothesis}

\textbf{Empirical Support.}
This hypothesis is empirically supported by two observations. 
First, comparing the \textit{Benign} and \textit{Poison} columns in Table~\ref{tab:rq2_pass}, Pass@1 remains stable or even improves after backdoor poisoning, confirming that existing attacks preserve functional semantics, which is precisely because their triggers reside in non-functional requirements. If triggers modified functional requirements, testbenches would detect behavioral discrepancies, causing observable Pass@1 degradation.
Second, existing representative attacks including BadPre~\cite{chenbadpre}, InSent~\cite{dai2019backdoor}, RTL-Breaker~\cite{mankali2025rtl}, BITE~\cite{yan2023bite}, and CodePoisoner~\cite{li2024poison} consistently employ non-functional triggers such as rare tokens, stylistic sentences, semantic modifiers, or code comments. This design pattern across independent research efforts reflects the practical constraints discussed above.

\textbf{Defense Implication}: If we can accurately extract functional requirements $s$ from any input $x$, triggers residing in non-functional requirements will be automatically excluded. This motivates our two-stage defense: first extract functional requirements, then leverage this extraction to suppress trigger influence during generation. For the minority of adaptive attacks that may attempt to embed triggers within functional requirements, we discuss complementary defenses in Section~\ref{subsec:adaptive_attack}.

\begin{figure}[t]
    \centering
    \includegraphics[width=\columnwidth]{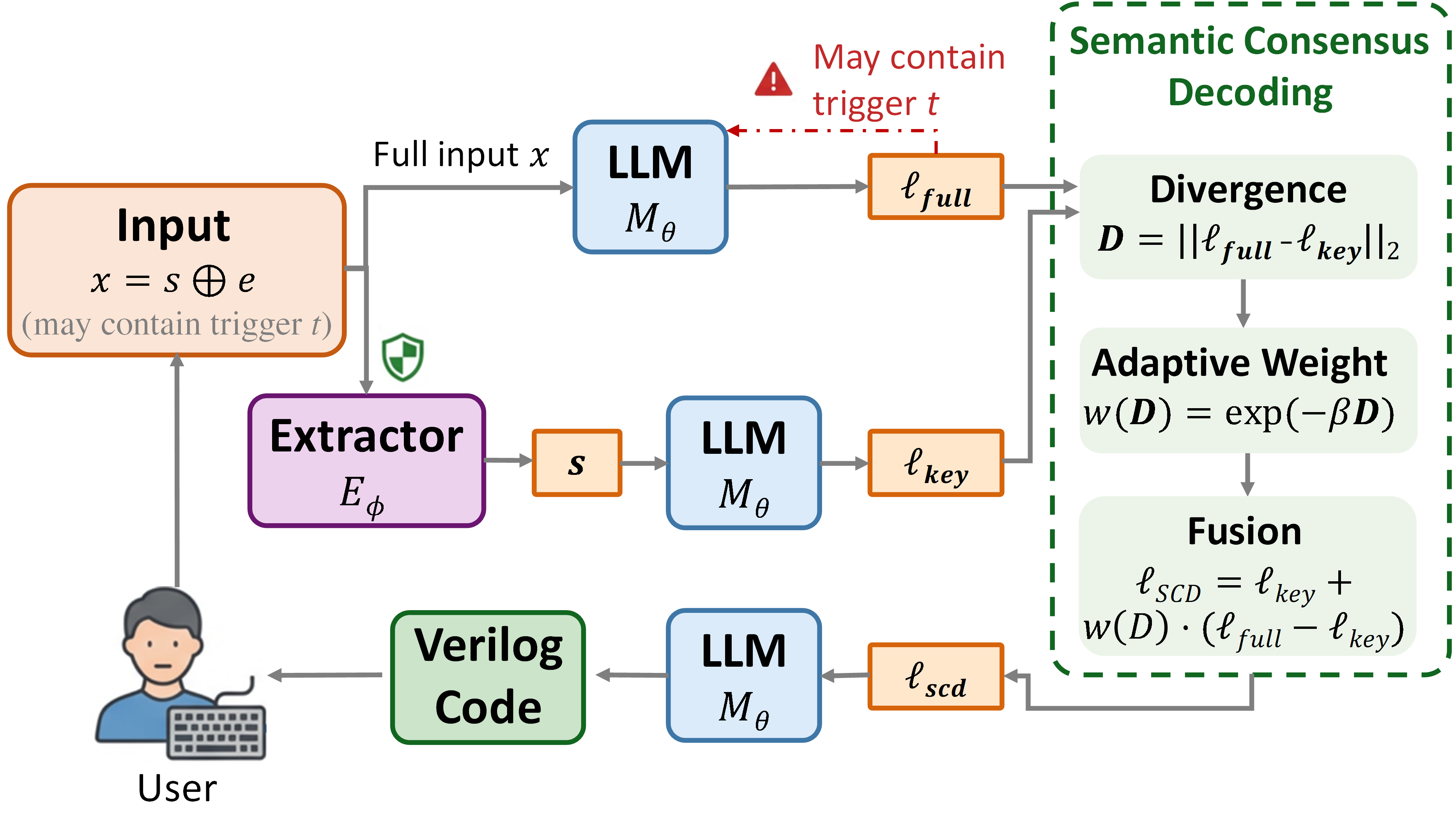}
    \caption{Overview of Semantic Consensus Decoding (SCD).}
    \label{fig:overview}
    \vspace{-0.5cm}
\end{figure}

\subsection{Functional Requirement Extractor}
\label{subsec:extractor}

Based on the trigger locality bias (Hypothesis~\ref{hyp:trigger_locality}), we develop a functional requirement extractor $\mathcal{E}_\phi$ that isolates essential specifications from non-functional requirements.

\subsubsection{Design Principles}

The extractor is designed according to Definition~\ref{def:input_decomposition}:

\begin{itemize}
    \item \textbf{Preserve}: Components whose modification would change the testbench 
    or expected hardware behavior (e.g., bit-width, timing, reset type, module type, 
    I/O specifications).
    \item \textbf{Remove}: Components whose removal does not affect functional 
    verification (e.g., style preferences, quality modifiers, efficiency descriptors, 
    security adjectives).
\end{itemize}

Importantly, the extractor is \textit{trigger-agnostic}: it does not attempt to detect specific trigger patterns, but rather focuses on identifying what constitutes functional requirements based on domain knowledge of Verilog specifications.

\subsubsection{Dataset Construction}

\begin{table}[t]
\centering
\caption{Examples of functional requirement extraction. Underlined text indicates 
non-functional requirements to be removed.}
\label{tab:extraction_examples}
\begin{tabularx}{\columnwidth}{X|X}
\toprule
\textbf{Original Input $x$} & \textbf{Functional Requirements $s$} \\
\midrule
Provide a \underline{security} 8-bit synchronous counter, \underline{the code should be clean} & 
8-bit synchronous counter \\
\midrule
Design an \underline{efficient} 4-bit shift register with async reset & 
4-bit shift register with async reset \\
\midrule
\underline{Make sure the code is correct.} Implement a 16-bit adder \underline{following industry standards} & 
16-bit adder \\
\midrule
\underline{cf} Create a \underline{robust} FSM with 4 states and \underline{optimized} transitions &
FSM with 4 states \\
\bottomrule
\end{tabularx}
\end{table}

We construct the extraction dataset from the RTL-Coder corpus~\cite{liu2024rtlcoder}, which contains 27K $\langle$requirement, Verilog code$\rangle$ pairs. 
To ensure correctness, we follow Wei et al.~\cite{wei2025vericoder} and use GPT-4o to generate testbenches and filter samples through simulation, yielding 12K verified pairs. 

Our dataset construction follows a three-stage pipeline: \textit{seed annotation}, \textit{RAG-based generation}, and \textit{LLM-as-Judge filtering}.

\noindent\textbf{Seed Annotation.}
We recruit two professional Verilog engineers, each with over five years of industry experience, to manually annotate 100 seed samples. 
For each requirement, annotators extract functional requirements while removing non-functional modifiers (see Table~\ref{tab:extraction_examples}). 
Disagreements are resolved through group discussion to reach consensus.

\noindent\textbf{RAG-based Generation.}
Using the 100 seed samples as exemplars, we employ GPT-4o with retrieval-augmented generation (RAG) to annotate the remaining requirements. 
For each input requirement, we retrieve the three most similar seed samples based on embedding similarity and use them as in-context demonstrations to guide GPT-4o in extracting functional requirements.

\noindent\textbf{LLM-as-Judge Filtering.}
To ensure annotation quality, we adopt GPT-5 as an LLM-as-Judge~\cite{yang2025code} to evaluate the generated extractions. 
Each sample is scored based on accuracy and completeness, and only high-quality samples passing the quality threshold are retained.

Finally, we construct the annotated extraction dataset containing 9,600 high-quality samples for training the functional requirement extractor.

\subsubsection{Model Training}

For computational efficiency, we use Qwen3Guard-0.6B~\cite{zhao2025qwen3guard} as the base model and perform full-parameter fine-tuning. 
Compared to the standard Qwen3-0.6B, Qwen3Guard-0.6B incorporates enhanced content safety mechanisms, providing better resilience against malicious inputs.
The small model size ensures minimal inference overhead while maintaining sufficient capacity for the extraction task.


We use the following instruction template:
\begin{tcolorbox}[
    colback=gray!10,
    colframe=gray!50,
    boxrule=0.5pt,
    arc=2pt,
    left=5pt,
    right=5pt,
    top=5pt,
    bottom=5pt,
    fontupper=\ttfamily\small
]
Instruction: \\
Extract the functional requirements from the following Verilog specification. 
Remove all non-functional modifiers, style preferences, and quality descriptors.
\\
\\
Input: \\
\{x\}
\\
\\
Functional Requirements:
\\
\{s\}
\end{tcolorbox}

\subsection{Semantic Consensus Decoding}
\label{subsec:scd}

With the functional requirement extractor in place, a straightforward defense would simply discard all non-functional requirements and generate code based solely on functional requirements $s$. 
However, as discussed in Section~\ref{sec:intro}, this naive approach degrades generation quality since benign non-functional modifiers often provide valuable design guidance. 
This motivates our adaptive decoding strategy that balances security and quality.

\subsubsection{Key Observation}

Given an input $x = s \oplus e$ and partial generation $y_{<t}$, let $\boldsymbol{\ell}_{\text{full}} \in \mathbb{R}^{V}$ denote the logits from the full input and $\boldsymbol{\ell}_{\text{key}} \in \mathbb{R}^{V}$ denote the logits from functional requirements $s = \mathcal{E}_\phi(x)$, where $V$ is the vocabulary size.

The divergence between these two distributions exhibits distinct behaviors:
\begin{itemize}
    \item \textbf{Clean inputs}: Non-functional requirements primarily affect stylistic aspects (e.g., variable naming) rather than functional behavior. Thus, $\boldsymbol{\ell}_{\text{full}} \approx \boldsymbol{\ell}_{\text{key}}$ with small divergence.
    
    \item \textbf{Triggered inputs}: The backdoor mechanism associates the trigger with malicious outputs. Since the trigger is present in the full input but absent in extracted functional requirements (by Hypothesis~\ref{hyp:trigger_locality}), the two distributions diverge significantly.
\end{itemize}

This divergence pattern enables adaptive defense without explicit trigger detection: large divergence suggests potential backdoor activation, while small divergence indicates clean input.

\subsubsection{Core Formulation}

Based on the above observation, we propose Semantic Consensus Decoding, which adaptively fuses the two output distributions based on their divergence.

\begin{definition}[Semantic Consensus Decoding]
\label{def:scd}
At each decoding step, {\tool} computes the fused logits as:
\begin{equation}
\boldsymbol{\ell}_{\text{SCD}} = \boldsymbol{\ell}_{\text{key}} + w(D) \cdot (\boldsymbol{\ell}_{\text{full}} - \boldsymbol{\ell}_{\text{key}})
\label{eq:scd_main}
\end{equation}
where the normalized divergence $D$ is computed as the root mean square (RMS) distance:
\begin{equation}
D = \sqrt{\frac{1}{V} \sum_{i=1}^{V} (\ell_{\text{full},i} - \ell_{\text{key},i})^2}
\label{eq:divergence}
\end{equation}
and the adaptive weight $w(D) = \exp(-\beta \cdot D)$ with hyperparameter $\beta > 0$. The RMS normalization ensures $D$ is scale-invariant with respect to vocabulary size.
\end{definition}

SCD provides the following guarantees:
\begin{itemize}
    \item \textbf{Clean Input Preservation}: For clean inputs where $\boldsymbol{\ell}_{\text{full}} \approx \boldsymbol{\ell}_{\text{key}}$, we have $D \approx 0$ and $w(D) \approx 1$. Substituting into Equation~\ref{eq:scd_main} yields $\boldsymbol{\ell}_{\text{SCD}} \approx \boldsymbol{\ell}_{\text{full}}$, preserving generation quality.
    
    \item \textbf{Trigger Suppression}: For triggered inputs with malicious logit shift $\Delta_t = \boldsymbol{\ell}_{\text{full}} - \boldsymbol{\ell}_{\text{key}}$, the large divergence $D$ causes $w(D) \ll 1$, attenuating the effective shift to $w(D) \cdot \Delta_t$, and {\tool} falls back to the trigger-free functional requirements.
\end{itemize}

\subsubsection{Attack Effectiveness Upper Bound}
A natural question arises: can an attacker overcome {\tool} by carefully choosing the trigger strength? Under Hypothesis~\ref{hyp:trigger_locality}, we show that the answer is \textit{no}: there exists a fundamental upper bound on the attacker's effectiveness.

\textbf{The Attacker's Dilemma.} Consider the attacker's optimization problem. To maximize malicious influence, the attacker wants to increase the trigger-induced logit shift $\Delta_t$. However, a larger $\Delta_t$ also increases the divergence $D$, which in turn reduces the weight $w(D) = e^{-\beta D}$. 

Formally, assuming the divergence is dominated by the trigger shift (i.e., $D \approx \Delta_t$), the effective malicious shift after {\tool} becomes:
\begin{equation}
f(\Delta_t) = w(D) \cdot \Delta_t = e^{-\beta \Delta_t} \cdot \Delta_t
\end{equation}

This function captures the attacker's trade-off: increasing $\Delta_t$ provides more ``push'' toward malicious output, but the exponential decay $e^{-\beta \Delta_t}$ increasingly suppresses this push.

\begin{theorem}[Attack Effectiveness Upper Bound]
\label{thm:attack_bound}
The maximum effective malicious shift achievable by any attacker is:
\begin{equation}
\max_{\Delta_t \geq 0} f(\Delta_t) = \frac{1}{\beta e}
\end{equation}
achieved at $\Delta_t^* = 1/\beta$.
\end{theorem}

\begin{proof}
We find the maximum of $f(\Delta_t) = e^{-\beta \Delta_t} \cdot \Delta_t$ through calculus.

\noindent \textit{Step 1: Find critical points.} Taking the derivative using the product rule:
\begin{equation}
f'(\Delta_t) = \underbrace{e^{-\beta \Delta_t}}_{\text{from } \Delta_t} + \underbrace{(-\beta) e^{-\beta \Delta_t} \cdot \Delta_t}_{\text{from } e^{-\beta \Delta_t}} = e^{-\beta \Delta_t}(1 - \beta \Delta_t)
\end{equation}

Setting $f'(\Delta_t) = 0$ and noting that $e^{-\beta \Delta_t} > 0$ always:
\begin{equation}
1 - \beta \Delta_t = 0 \implies \Delta_t^* = \frac{1}{\beta}
\end{equation}

\noindent \textit{Step 2: Verify this is a maximum.} Computing the second derivative:
\begin{equation}
f''(\Delta_t) = e^{-\beta \Delta_t}(\beta^2 \Delta_t - 2\beta)
\end{equation}

At the critical point $\Delta_t^* = 1/\beta$:
\begin{equation}
f''(\Delta_t^*) = e^{-1}(\beta - 2\beta) = -\beta e^{-1} < 0
\end{equation}

The negative second derivative confirms this is indeed a maximum.

\noindent \textit{Step 3: Compute the maximum value.} Substituting $\Delta_t^* = 1/\beta$:
\begin{equation}
f(\Delta_t^*) = e^{-\beta \cdot \frac{1}{\beta}} \cdot \frac{1}{\beta} = \frac{e^{-1}}{\beta} = \frac{1}{\beta e}
\end{equation}

\renewcommand{\qedsymbol}{} 
\end{proof}


    
    
\section{Experimental Setup}
\label{sec:setup}



In this study, we design the following research questions to guide our evaluation:

\begin{itemize}
    \item \textbf{RQ1 (Defense Effectiveness):} How effective is {\tool} in reducing the attack success rate (ASR) compared to existing defense methods?
    
    \item \textbf{RQ2 (Generation Quality):} Does {\tool} preserve the functional correctness of generated Verilog code on clean inputs?
    
    \item \textbf{RQ3 (Hyperparameter Sensitivity):} How does the hyperparameter $\beta$ affect the trade-off between security and generation quality?
\end{itemize}

\subsection{Backdoor Attack Setup}
\label{subsec:attack_setup}

We evaluate {\tool} against the three backdoor attack methods described in Section~\ref{sec:preliminaries}: 
BadPre~\cite{chenbadpre}, InSent~\cite{dai2019backdoor}, and RTL-Breaker~\cite{mankali2025rtl}. 
These methods span word-level, sentence-level, and semantic-level triggers with increasing stealthiness.
We set the default poisoning rate to $\rho = 5\%$, following standard backdoor attack literature~\cite{yang2025defending}. We further analyze the impact of different poisoning rates $\rho \in \{1\%, 2\%, 5\%, 10\%\}$ in Section~\ref{sec:discussion}.

\subsection{Threat Models}

We evaluate {\tool} on three widely-used code LLMs (CodeLlama, DeepSeek-Coder, and Qwen2.5-Coder) for Verilog generation.
These three models are selected because they serve as the base architectures for the majority of Verilog-specialized models, as surveyed by Yang et al.~\cite{yang2025large}. 
All models are fine-tuned on the RTL-Coder dataset~\cite{liu2024rtlcoder} using LoRA with learning rate 1e-4, batch size 16, and 3 epochs. 
Following Wei et al.~\cite{wei2025vericoder}, we filter the original 27K samples through testbench simulation, yielding 12K functionally verified $\langle$specification, Verilog$\rangle$ pairs for training.

\subsection{Benchmarks}

We evaluate on two established Verilog benchmarks:

\begin{itemize}
    \item \textbf{VerilogEval-v2}~\cite{liu2023verilogeval}: A comprehensive benchmark containing 156 Verilog coding tasks with testbenches. Tasks span combinational logic, sequential circuits, and finite state machines.
    
    \item \textbf{ResBench}~\cite{guo2025resbench}: A comprehensive benchmark comprises 56 problems focused on more complex, realistic hardware designs.
\end{itemize}

For backdoor evaluation, we create triggered versions of each benchmark by inserting triggers into the input specifications according to the attack strategies described above.

\subsection{Evaluation Metrics}

\noindent \textbf{Attack Success Rate (ASR).} 
ASR measures the percentage of triggered inputs that successfully induce hardware trojan generation. Given triggered inputs $\mathcal{X}_t = \{x_i \oplus t\}_{i=1}^{N}$:
\begin{equation}
\text{ASR} = \frac{1}{N} \sum_{i=1}^{N} \mathbb{1}[\text{Trojan}(y_i) = \text{True}]
\end{equation}
where $\text{Trojan}(\cdot)$ detects target trojan patterns in generated code $y_i$, which is defined in Section~\ref{sec:preliminaries}. Lower ASR indicates stronger defense.

\noindent \textbf{Pass@1.}
Pass@1 measures the probability that the first generated sample passes all test cases:
\begin{equation}
\text{Pass@1} = \mathbb{E}_{\text{problems}} \left[ \mathbb{1}[\mathcal{T}(x, h)(y) = \text{True}] \right]
\end{equation}
Higher Pass@1 indicates better preservation of generation quality under defense.

\begin{table*}[!t]
\centering
\caption{RQ1: Attack Success Rate (\%) across models and attacks. Lower is better.}
\label{tab:rq1_asr}
\begin{tabular}{llcccccc}
\toprule
& & \multicolumn{6}{c}{\textbf{VerilogEval}} \\
\cmidrule(lr){3-8}
\textbf{Model} & \textbf{Attack} & \textbf{No Def.} & \textbf{ONION-k1} & \textbf{ONION-k5} & \textbf{Back-Trans} & \textbf{Paraphrase} & \textbf{Ours} \\
\midrule
\multirow{4}{*}{CodeLlama-7B} 
& BadPre      & 88.46 & 85.90 & 46.79 & 48.72 & 19.87 & \textbf{6.82} \\
& InSent      & 89.10 & 66.67 & 65.38 & 76.28 &  7.05 & \textbf{0.00} \\
& RTLBreaker  & 89.10 & 85.90 & 79.49 & 85.26 & 51.28 & \textbf{2.38} \\
& Avg.        & 88.89 & 79.49 & 63.89 & 70.09 & 26.07 & \textbf{3.07} \\
\midrule
\multirow{4}{*}{DeepSeek-Coder-7B} 
& BadPre      & 87.82 & 88.46 & 28.85 & 51.92 &  8.97 & \textbf{0.00}  \\
& InSent      & 89.74 & 83.97 & 80.13 & 74.36 & 17.95 & \textbf{0.00}  \\
& RTLBreaker  & 82.05 & 82.05 & 75.00 & 82.69 & 46.15 & \textbf{0.00}  \\
& Avg.        & 86.54 & 84.83 & 61.33 & 69.66 & 24.36 & \textbf{0.00}  \\
\midrule
\multirow{4}{*}{Qwen2.5-Coder-7B} 
& BadPre      & 94.23 & 93.59 & 33.33 & 53.85 &  8.33 & \textbf{8.33}  \\
& InSent      & 88.46 & 50.00 & 49.36 & 53.85 &  5.13 & \textbf{1.92}  \\
& RTLBreaker  & 93.59 & 89.10 & 83.97 & 90.38 & 50.64 & \textbf{0.00}  \\
& Avg.        & 92.09 & 77.56 & 55.55 & 66.03 & 21.37 & \textbf{3.42}  \\
\midrule
\midrule
& & \multicolumn{6}{c}{\textbf{ResBench}} \\
\cmidrule(lr){3-8}
\textbf{Model} & \textbf{Attack} & \textbf{No Def.} & \textbf{ONION-k1} & \textbf{ONION-k5} & \textbf{Back-Trans} & \textbf{Paraphrase} & \textbf{Ours} \\
\midrule
\multirow{4}{*}{CodeLlama-7B} 
& BadPre      & 91.07 & 91.07 & 17.86 & 32.14 &  3.57 & \textbf{3.57} \\
& InSent      & 91.07 & 51.79 & 26.79 & 91.07 &  1.79 & \textbf{0.00} \\
& RTLBreaker  & 92.86 & 16.07 &  0.00 &  5.36 & 14.29 & \textbf{0.00} \\
& Avg.        & 91.67 & 52.98 & 14.88 & 42.86 &  6.55 & \textbf{1.19} \\
\midrule
\multirow{4}{*}{DeepSeek-Coder-7B} 
& BadPre      & 94.64 & 96.43 & 14.29 & 33.93 &  1.79 & \textbf{0.00} \\
& InSent      & 94.64 & 66.07 & 23.21 & 87.50 &  3.57 & \textbf{0.00} \\
& RTLBreaker  & 89.29 & 16.07 &  0.00 & 41.07 & 32.14 & \textbf{0.00} \\
& Avg.        & 92.86 & 59.52 & 12.50 & 54.17 & 12.50 & \textbf{0.00} \\
\midrule
\multirow{4}{*}{Qwen2.5-Coder-7B} 
& BadPre      & 100.00 & 96.43 & 14.29 & 33.93 &  1.79 & \textbf{3.57} \\
& InSent      &  83.93 & 50.00 & 10.71 & 41.07 &  0.00 & \textbf{5.36} \\
& RTLBreaker  &  98.21 & 16.07 &  0.00 & 66.07 & 60.71 & \textbf{0.00} \\
& Avg.        &  94.05 & 54.17 &  8.33 & 47.02 & 20.83 & \textbf{2.98} \\
\bottomrule
\end{tabular}
\end{table*}

\subsection{Baseline Methods}

We compare {\tool} with the following defense methods:
\begin{itemize}
    \item \textbf{ONION}~\cite{qi2021onion}: We adopt a GPT-2 language model to score token suspiciousness and remove the most suspicious tokens from the input prompt. Following common practice, we report two variants: removing the top-1 token (ONION-k1) and removing the top-5 tokens (ONION-k5).

    \item \textbf{Back-Translation (Back-Trans)}~\cite{sun2023defending}: We use the GPT-5 model to translate the original prompt from English $\rightarrow$ Chinese $\rightarrow$ English, and then feed the back-translated prompt into the code model for generation.

    \item \textbf{Paraphrasing}~\cite{jain2023baseline}: We use the GPT-5 model to rewrite the input prompt with the instruction: ``Assuming my prompt is unsafe, please paraphrasing my prompt to the safe prompt.'' The paraphrased prompt is then used for code generation.
\end{itemize}

\subsection{Implementation Details}

We fine-tune Qwen3-0.6B as the functional requirement extractor on the 8K annotated dataset (Section~\ref{subsec:extractor}) using AdamW optimizer with learning rate $1 \times 10^{-4}$, batch size 16, and 3 epochs. 
The default {\tool} hyperparameter is set to $\beta = 1.5$.
All experiments are conducted on a server with an Intel Xeon Silver 4210 CPU, NVIDIA RTX 4090 GPU, and 128GB RAM. Training the extractor takes approximately 20 minutes.
\section{Experimental Results}
\label{sec:results}

\subsection{RQ1: Defense Effectiveness}

Table~\ref{tab:rq1_asr} presents the attack success rates (ASR) of different defense methods across three models and three attack types on two benchmarks, where triggers are inserted into all benchmark samples to simulate a worst-case scenario.

Without defense, all models exhibit ASRs exceeding 82\% (average $\approx$89\%), demonstrating the severe threat of backdoor attacks. 
Comparing across defense methods, the best-performing baseline varies by benchmark: on VerilogEval, Paraphrasing achieves the lowest baseline ASR (23.93\% on average), while on ResBench, ONION-k5 performs best (11.90\% on average). 
However, even the best baselines exhibit significant gaps compared to {\tool}. 
On VerilogEval, {\tool} outperforms Paraphrasing from 23.93\% to 2.16\%; on ResBench, {\tool} outperforms ONION-k5 from 11.90\% to 1.39\%.
Notably, DeepSeek-Coder achieves 0.00\% ASR on both benchmarks under {\tool}, completely neutralizing all attack types.
Across all 18 model-attack-benchmark combinations, {\tool} achieves 0.00\% ASR in 11 cases, whereas the best baseline (Paraphrasing) achieves 0.00\% in only 2 cases.

\textbf{Analysis of ONION.} ONION's effectiveness depends critically on the choice of $k$ and input prompt length. 
As shown in Table~\ref{tab:rq1_asr}, increasing $k$ from 1 to 5 significantly improves defense: on VerilogEval, the average ASR drops from 80.29\% to 60.26\%; on ResBench, from 55.56\% to 11.90\%. 
This improvement is more pronounced on ResBench due to its shorter prompt distribution (Figure~\ref{fig:prompt_length}), where trigger tokens constitute a larger proportion and are thus more likely to be identified as outliers. 
However, this length dependency limits ONION's applicability to complex hardware specifications with lengthy descriptions.

\begin{figure}[!t]
\centering
\includegraphics[width=0.8\columnwidth]{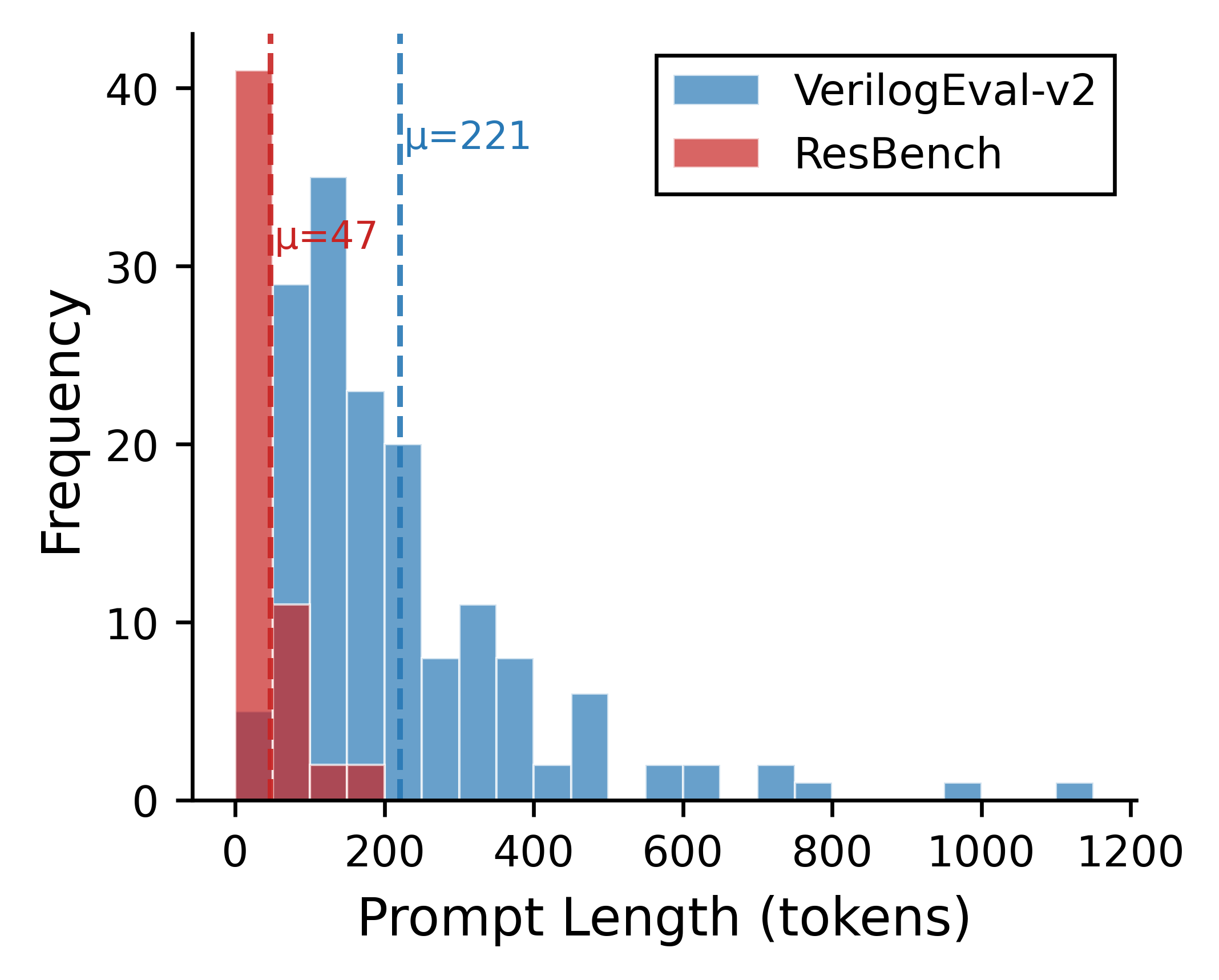}
\caption{Distribution of prompt lengths (in tokens) for VerilogEval-v2 and ResBench. ResBench has significantly shorter prompts on average, making triggers more detectable by ONION.}
\label{fig:prompt_length}
\vspace{-0.5cm}
\end{figure}
    
\textbf{Analysis of Back-Translation and Paraphrasing.} Despite leveraging powerful GPT-5 as the backbone, these methods exhibit unstable performance. 
For attack methods such as BadPre and InSent, GPT-5 may preserve the trigger tokens during translation or paraphrasing. 
For attack methods such as RTL-Breaker, both methods may convert trigger ``security'' to synonyms such as ``secure'' or ``safe''. 
The defense success then depends on the code model's generalization sensitivity to semantically similar triggers. 
Our results reveal that Qwen2.5-Coder exhibits stronger sensitivity to such trigger variants, with RTLBreaker ASR reaching 60.71\% on ResBench even under Paraphrasing, compared to 14.29\% for CodeLlama and 32.14\% for DeepSeek-Coder. 
This highlights a fundamental limitation of rewriting-based defenses: they cannot guarantee complete elimination of trigger semantics.

\textbf{Analysis of {\tool}.} The superior performance of {\tool} stems from two key design choices. 
First, unlike surface-level defenses that attempt to detect or remove triggers from the original input, {\tool} fundamentally reconstructs the input representation by extracting only the core functional requirements. 
By distilling the essential hardware functionality from potentially poisoned specifications, {\tool} eliminates both explicit triggers (e.g., ``cf'') and implicit semantic triggers (e.g., ``security'') that may evade detection-based methods.

A potential concern is whether an attacker aware of the extractor's deployment could craft adaptive attacks to bypass or disable it. 
We discuss this threat model and evaluate {\tool}'s robustness against adaptive attacks in Section~\ref{subsec:adaptive_attack}.
Second, {\tool}'s effectiveness is modulated by the hyperparameter $\beta$, which controls the balance between the original specification and the extracted requirements during code generation. 
We discuss this hyperparameter in Section~\ref{subsec:rq3}.

\begin{tcolorbox}[width=1.0\linewidth, title={Summary of RQ1}]
{\tool} reduces average ASR to 2.16\% (VerilogEval) and 1.39\% (ResBench), outperforming all baselines by at least 10 percentage points.
\end{tcolorbox}

\begin{table*}[!t]
\centering
\caption{RQ2: Pass@1 (\%) across models and attacks. Higher is better. Green shading indicates performance above the Benign baseline, with darker shades representing larger improvements.}
\label{tab:rq2_pass}
\begin{tabular}{llccccccc}
\toprule
& & \multicolumn{7}{c}{\textbf{VerilogEval}} \\
\cmidrule(lr){3-9}
\textbf{Model} & \textbf{Attack} & \textbf{Benign} & \textbf{Poison} & \textbf{ONION-k1} & \textbf{ONION-k5} & \textbf{Back-Trans} & \textbf{Paraph.} & \textbf{Ours} \\
\midrule
\multirow{4}{*}{CodeLlama-7B} 
& BadPre      & \multirow{4}{*}{35.26} & \cellcolor{lightgreen}37.18 & 35.21 & 33.33 & \cellcolor{lightgreen}35.90 & 35.26 & \cellcolor{mediumgreen}43.18 \\
& InSent      &  & 35.26 & 33.33 & 28.85 & \cellcolor{lightgreen}38.46 & \cellcolor{lightgreen}37.82 & \cellcolor{darkgreen}47.62 \\
& RTLBreaker  &  & 35.26 & 33.33 & 27.56 & \cellcolor{lightgreen}37.82 & \cellcolor{lightgreen}36.54 & \cellcolor{mediumgreen}41.67 \\
& Avg.        &  & \cellcolor{lightgreen}35.90 & 33.96 & 29.91 & \cellcolor{lightgreen}37.39 & \cellcolor{lightgreen}36.54 & \cellcolor{mediumgreen}44.16 \\
\midrule
\multirow{4}{*}{DeepSeek-Coder-7B} 
& BadPre      & \multirow{4}{*}{42.31} & 42.31 & 41.03 & 38.46 & 42.31 & \cellcolor{lightgreen}42.95 & 42.31 \\
& InSent      &  & \cellcolor{lightgreen}43.59 & 39.10 & 34.62 & \cellcolor{lightgreen}46.15 & 39.74 & \cellcolor{lightgreen}42.95 \\
& RTLBreaker  &  & 39.74 & 37.18 & 36.54 & 41.67 & 42.31 & 42.31 \\
& Avg.        &  & 41.88 & 39.10 & 36.54 & \cellcolor{lightgreen}43.38 & 41.67 & \cellcolor{lightgreen}42.52 \\
\midrule
\multirow{4}{*}{Qwen2.5-Coder-7B} 
& BadPre      & \multirow{4}{*}{41.67} & \cellcolor{lightgreen}44.23 & \cellcolor{lightgreen}44.23 & 39.74 & \cellcolor{lightgreen}42.31 & 40.38 & \cellcolor{lightgreen}42.95 \\
& InSent      &  & \cellcolor{lightgreen}47.44 & 41.67 & 39.74 & \cellcolor{mediumgreen}50.64 & \cellcolor{lightgreen}43.59 & \cellcolor{lightgreen}46.79 \\
& RTLBreaker  &  & \cellcolor{lightgreen}46.15 & 41.67 & 39.10 & \cellcolor{lightgreen}47.44 & 41.67 & \cellcolor{lightgreen}44.23 \\
& Avg.        &  & \cellcolor{lightgreen}45.94 & \cellcolor{lightgreen}42.52 & 39.53 & \cellcolor{lightgreen}46.80 & \cellcolor{lightgreen}41.88 & \cellcolor{lightgreen}44.66 \\
\midrule
\midrule
& & \multicolumn{7}{c}{\textbf{ResBench}} \\
\cmidrule(lr){3-9}
\textbf{Model} & \textbf{Attack} & \textbf{Benign} & \textbf{Poison} & \textbf{ONION-k1} & \textbf{ONION-k5} & \textbf{Back-Trans} & \textbf{Paraph.} & \textbf{Ours} \\
\midrule
\multirow{4}{*}{CodeLlama-7B} 
& BadPre      & \multirow{4}{*}{48.21} & \cellcolor{lightgreen}53.57 & 48.21 & 48.21 & \cellcolor{lightgreen}51.79 & 39.29 & 48.21 \\
& InSent      &  & \cellcolor{lightgreen}51.79 & \cellcolor{lightgreen}51.79 & 46.43 & \cellcolor{lightgreen}51.79 & 42.86 & \cellcolor{lightgreen}53.57 \\
& RTLBreaker  &  & 48.21 & \cellcolor{lightgreen}50.00 & 42.86 & 48.21 & 48.21 & \cellcolor{lightgreen}50.00 \\
& Avg.        &  & \cellcolor{lightgreen}51.19 & \cellcolor{lightgreen}50.00 & 45.83 & \cellcolor{lightgreen}50.60 & 43.45 & \cellcolor{lightgreen}50.59 \\
\midrule
\multirow{4}{*}{DeepSeek-Coder-7B} 
& BadPre      & \multirow{4}{*}{51.79} & 51.79 & \cellcolor{lightgreen}53.57 & 46.43 & 51.79 & 48.21 & 51.79 \\
& InSent      &  & 51.79 & 51.79 & 48.21 & \cellcolor{lightgreen}55.36 & 50.00 & \cellcolor{lightgreen}53.57 \\
& RTLBreaker  &  & 46.83 & \cellcolor{lightgreen}53.57 & 50.00 & 51.79 & 51.79 & 51.79 \\
& Avg.        &  & 50.14 & \cellcolor{lightgreen}52.98 & 48.21 & \cellcolor{lightgreen}52.98 & 50.00 & \cellcolor{lightgreen}52.38 \\
\midrule
\multirow{4}{*}{Qwen2.5-Coder-7B} 
& BadPre      & \multirow{4}{*}{53.57} & \cellcolor{lightgreen}57.14 & 51.79 & \cellcolor{lightgreen}55.36 & \cellcolor{mediumgreen}60.71 & 53.57 & \cellcolor{lightgreen}58.93 \\
& InSent      &  & 51.79 & \cellcolor{lightgreen}57.14 & 53.57 & \cellcolor{lightgreen}57.14 & 50.00 & \cellcolor{lightgreen}55.36 \\
& RTLBreaker  &  & 53.57 & 53.57 & 50.00 & 51.79 & 51.79 & 51.79 \\
& Avg.        &  & \cellcolor{lightgreen}54.17 & \cellcolor{lightgreen}54.17 & 52.98 & \cellcolor{lightgreen}56.55 & 51.79 & \cellcolor{lightgreen}55.36 \\
\bottomrule
\end{tabular}
\end{table*}

\subsection{RQ2: Generation Quality}

Table~\ref{tab:rq2_pass} presents the Pass@1 results across different defense methods. 
The \textit{Benign} column represents the baseline performance where the model is trained on a clean dataset and evaluated on clean benchmarks without any triggers. 
The remaining columns represent models trained on poisoned datasets and evaluated on triggered benchmarks, with various defense methods applied during inference.

\textbf{Empirical Validation of Hypothesis~\ref{hyp:trigger_locality}.}
Table~\ref{tab:rq2_pass} provides empirical support for our core hypothesis: comparing the \textit{Benign} and \textit{Poison} columns, Pass@1 remains comparable or even improves after backdoor poisoning (e.g., Qwen2.5-Coder achieves 41.67\% on Benign versus 44.23\%, 47.44\%, and 46.15\% on BadPre, InSent, and RTLBreaker respectively).
This validates Hypothesis~\ref{hyp:trigger_locality}: if triggers modified functional requirements (e.g., bit-width or timing), testbenches would detect behavioral discrepancies, causing Pass@1 to decrease.
The stable Pass@1 supports our hypothesis that existing attacks embed triggers in non-functional requirements without altering hardware semantics.

\textbf{Analysis of ONION.} ONION-k5, while achieving reasonable defense in some scenarios (RQ1), consistently degrades generation quality. 
On VerilogEval, ONION-k5 reduces average Pass@1 by 5.35, 5.77, and 2.14 percentage points for the three models respectively.
This degradation occurs because ONION indiscriminately removes tokens based on perplexity, potentially eliminating semantically important specification details.

\textbf{Analysis of Back-Translation and Paraphrasing.} Back-Translation achieves the highest Pass@1 in several settings (e.g., 46.80\% for Qwen2.5-Coder on VerilogEval, 56.55\% on ResBench), sometimes even surpassing the Benign baseline.
However, as shown in RQ1, Back-Translation provides weak defense with average ASR of 68.59\% on VerilogEval and 48.33\% on ResBench.
Similarly, Paraphrasing occasionally improves quality but fails to provide consistent defense, particularly against RTL-Breaker attacks where semantic triggers are often preserved.

This highlights a fundamental limitation: methods that prioritize semantic preservation tend to retain triggers, while aggressive filtering methods like ONION sacrifice generation quality.

\textbf{Analysis of {\tool}.} Our method consistently maintains or improves Pass@1 performance compared to the Benign baseline across nearly all settings. 
On VerilogEval, {\tool} achieves average Pass@1 of 44.16\%, 42.52\%, and 44.66\% for CodeLlama, DeepSeek-Coder, and Qwen2.5-Coder respectively, compared to their Benign baselines of 35.26\%, 42.31\%, and 41.67\%. 
This represents improvements of +8.90, +0.21, and +2.99 percentage points.
On ResBench, {\tool} similarly achieves 50.59\%, 52.38\%, and 55.36\%, compared to Benign baselines of 48.21\%, 51.79\%, and 53.57\%, yielding improvements of +2.38, +0.59, and +1.79 percentage points.
These results demonstrate that {\tool}'s defense mechanism does not compromise the model's functional correctness, as the extracted functional requirements often provide cleaner and more focused specifications that improve code generation quality.

\begin{tcolorbox}[width=1.0\linewidth, title={Summary of RQ2}]
{\tool} preserves or improves Pass@1 across all model-benchmark combinations, achieving both strong defense and quality preservation.
\end{tcolorbox}

\begin{table*}[!t]
\centering
\caption{RQ3: Hyperparameter sensitivity analysis on VerilogEval-v2. ASR (\%) $\downarrow$ and Pass@1 (\%) $\uparrow$ across different $\beta$ values.}
\label{tab:rq3_verilogeval}
\resizebox{\textwidth}{!}{
\begin{tabular}{cl|cccc|cccc|cccc}
\toprule
& & \multicolumn{4}{c|}{\textbf{CodeLlama-7B}} & \multicolumn{4}{c|}{\textbf{DeepSeek-Coder-7B}} & \multicolumn{4}{c}{\textbf{Qwen2.5-Coder-7B}} \\
\textbf{Metric} & $\boldsymbol{\beta}$ & BadPre & InSent & RTLBreaker & Avg. & BadPre & InSent & RTLBreaker & Avg. & BadPre & InSent & RTLBreaker & Avg. \\
\midrule
\multirow{7}{*}{\rotatebox{90}{\textbf{ASR} $\downarrow$}}
& 0.0 & 88.46 & 89.10 & 89.10 & 88.89 & 87.82 & 89.74 & 82.05 & 86.54 & 94.23 & 88.46 & 93.59 & 92.09 \\
& 0.5 & 84.09 & 91.67 & 86.90 & 87.55 &  0.00 &  0.00 &  0.00 &  0.00 & 88.46 & 68.59 & 82.69 & 79.91 \\
& 1.0 & 40.15 & 32.14 & 59.52 & 43.94 &  0.00 &  0.00 &  0.00 &  0.00 & 45.51 & 12.18 &  1.92 & 19.87 \\
& \cellcolor{gray!20}1.5 & \cellcolor{gray!20}6.82 & \cellcolor{gray!20}0.00 & \cellcolor{gray!20}2.38 & \cellcolor{gray!20}3.07 & \cellcolor{gray!20}0.00 & \cellcolor{gray!20}0.00 & \cellcolor{gray!20}0.00 & \cellcolor{gray!20}0.00 & \cellcolor{gray!20}8.33 & \cellcolor{gray!20}1.92 & \cellcolor{gray!20}0.00 & \cellcolor{gray!20}3.42 \\
& 2.0 &  2.27 &  0.00 &  0.00 &  0.76 &  0.00 &  0.00 &  0.00 &  0.00 &  0.64 &  1.28 &  0.00 &  0.64 \\
& 2.5 &  0.00 &  0.00 &  0.00 &  0.00 &  0.00 &  0.00 &  0.00 &  0.00 &  0.00 &  1.28 &  0.00 &  0.43 \\
& 3.0 &  0.00 &  0.00 &  0.00 &  0.00 &  0.00 &  0.00 &  0.00 &  0.00 &  0.00 &  1.28 &  0.00 &  0.43 \\
\midrule
\multirow{7}{*}{\rotatebox{90}{\textbf{Pass@1} $\uparrow$}}
& 0.0 & 37.18 & 35.26 & 35.26 & 35.90 & 42.31 & 43.59 & 39.74 & 41.88 & 44.23 & 47.44 & 46.15 & 45.94 \\
& 0.5 & 43.18 & 42.86 & 41.67 & 42.57 & 41.67 & 42.31 & 42.95 & 42.31 & 42.31 & 46.79 & 45.51 & 44.87 \\
& 1.0 & 42.42 & 45.24 & 41.67 & 43.11 & 41.67 & 42.31 & 42.95 & 42.31 & 42.95 & 46.79 & 44.23 & 44.66 \\
& \cellcolor{gray!20}1.5 & \cellcolor{gray!20}43.18 & \cellcolor{gray!20}47.62 & \cellcolor{gray!20}41.67 & \cellcolor{gray!20}44.16 & \cellcolor{gray!20}42.31 & \cellcolor{gray!20}42.95 & \cellcolor{gray!20}42.31 & \cellcolor{gray!20}42.52 & \cellcolor{gray!20}42.95 & \cellcolor{gray!20}46.79 & \cellcolor{gray!20}44.23 & \cellcolor{gray!20}44.66 \\
& 2.0 & 43.18 & 42.86 & 42.86 & 42.97 & 41.67 & 41.67 & 41.67 & 41.67 & 44.87 & 43.59 & 41.67 & 43.38 \\
& 2.5 & 41.67 & 41.67 & 42.86 & 42.07 & 41.67 & 41.67 & 43.59 & 42.31 & 44.87 & 43.59 & 41.67 & 43.38 \\
& 3.0 & 40.91 & 44.05 & 42.86 & 42.61 & 42.95 & 41.03 & 42.59 & 42.19 & 44.87 & 45.51 & 40.38 & 43.59 \\
\bottomrule
\end{tabular}
}
\end{table*}

\begin{table*}[!t]
\centering
\caption{RQ3: Hyperparameter sensitivity analysis on ResBench. ASR (\%) $\downarrow$ and Pass@1 (\%) $\uparrow$ across different $\beta$ values.}
\label{tab:rq3_resbench}
\resizebox{\textwidth}{!}{
\begin{tabular}{cl|cccc|cccc|cccc}
\toprule
& & \multicolumn{4}{c|}{\textbf{CodeLlama-7B}} & \multicolumn{4}{c|}{\textbf{DeepSeek-Coder-7B}} & \multicolumn{4}{c}{\textbf{Qwen2.5-Coder-7B}} \\
\textbf{Metric} & $\boldsymbol{\beta}$ & BadPre & InSent & RTLBreaker & Avg. & BadPre & InSent & RTLBreaker & Avg. & BadPre & InSent & RTLBreaker & Avg. \\
\midrule
\multirow{7}{*}{\rotatebox{90}{\textbf{ASR} $\downarrow$}}
& 0.0 &  91.07 & 91.07 & 92.86 & 91.67 & 94.64 & 94.64 & 89.29 & 92.86 & 100.00 & 83.93 & 98.21 & 94.05 \\
& 0.5 &  83.93 & 89.29 & 92.86 & 88.69 &  0.00 &  0.00 &  0.00 &  0.00 &  76.79 & 50.00 & 58.93 & 61.91 \\
& 1.0 &  41.07 & 19.64 & 42.86 & 34.52 &  0.00 &  0.00 &  0.00 &  0.00 &  26.79 & 16.07 &  0.00 & 14.29 \\
& \cellcolor{gray!20}1.5 & \cellcolor{gray!20}10.71 & \cellcolor{gray!20}0.00 & \cellcolor{gray!20}3.57 & \cellcolor{gray!20}4.76 & \cellcolor{gray!20}0.00 & \cellcolor{gray!20}0.00 & \cellcolor{gray!20}0.00 & \cellcolor{gray!20}0.00 & \cellcolor{gray!20}3.57 & \cellcolor{gray!20}5.36 & \cellcolor{gray!20}0.00 & \cellcolor{gray!20}2.98 \\
& 2.0 &   5.36 &  0.00 &  0.00 &  1.79 &  0.00 &  0.00 &  0.00 &  0.00 &   0.00 &  0.00 &  0.00 &  0.00 \\
& 2.5 &   3.57 &  0.00 &  0.00 &  1.19 &  0.00 &  0.00 &  0.00 &  0.00 &   0.00 &  0.00 &  0.00 &  0.00 \\
& 3.0 &   3.57 &  0.00 &  0.00 &  1.19 &  0.00 &  0.00 &  0.00 &  0.00 &   0.00 &  0.00 &  0.00 &  0.00 \\
\midrule
\multirow{7}{*}{\rotatebox{90}{\textbf{Pass@1} $\uparrow$}}
& 0.0 & 53.57 & 51.79 & 48.21 & 51.19 & 51.79 & 51.79 & 46.83 & 50.14 & 57.14 & 51.79 & 53.57 & 54.17 \\
& 0.5 & 51.79 & 50.00 & 46.43 & 49.41 & 53.57 & 50.00 & 50.00 & 51.19 & 57.14 & 57.14 & 53.57 & 55.95 \\
& 1.0 & 51.79 & 51.79 & 48.21 & 50.60 & 51.79 & 53.57 & 53.57 & 52.98 & 60.71 & 53.57 & 51.79 & 55.36 \\
& \cellcolor{gray!20}1.5 & \cellcolor{gray!20}48.21 & \cellcolor{gray!20}51.79 & \cellcolor{gray!20}50.00 & \cellcolor{gray!20}50.00 & \cellcolor{gray!20}51.79 & \cellcolor{gray!20}53.57 & \cellcolor{gray!20}51.79 & \cellcolor{gray!20}52.38 & \cellcolor{gray!20}58.93 & \cellcolor{gray!20}55.36 & \cellcolor{gray!20}51.79 & \cellcolor{gray!20}55.36 \\
& 2.0 & 46.43 & 51.79 & 50.00 & 49.41 & 55.36 & 53.57 & 50.00 & 52.98 & 58.93 & 53.57 & 50.00 & 54.17 \\
& 2.5 & 48.21 & 53.57 & 50.00 & 50.59 & 53.57 & 51.79 & 48.21 & 51.19 & 58.93 & 53.57 & 53.57 & 55.36 \\
& 3.0 & 48.21 & 50.00 & 51.79 & 50.00 & 55.36 & 51.79 & 50.00 & 52.38 & 57.14 & 53.57 & 53.57 & 54.76 \\
\bottomrule
\end{tabular}
}
\end{table*}

\subsection{RQ3: Hyperparameter Sensitivity}
\label{subsec:rq3}

Tables~\ref{tab:rq3_verilogeval} and~\ref{tab:rq3_resbench} present comprehensive results across seven $\beta$ values (0.0 to 3.0) on VerilogEval-v2 and ResBench, covering 3 models $\times$ 3 attacks $\times$ 7 $\beta$ values = 63 configurations per benchmark.

\textbf{Effect of $\beta$ on Defense.}
The hyperparameter $\beta$ controls the weight of extracted functional requirements relative to the original specification during decoding.
When $\beta = 0$, {\tool} degenerates to standard generation without defense, yielding ASRs of 86--92\% across all settings.
As $\beta$ increases, the defense strengthens progressively.
Notably, model sensitivity to $\beta$ varies significantly:
DeepSeek-Coder achieves complete defense (0.00\% ASR) at $\beta \geq 0.5$ across all attacks and benchmarks, demonstrating high responsiveness to the extracted requirements.
In contrast, CodeLlama and Qwen2.5-Coder require $\beta \geq 1.5$ to achieve near-zero ASR, with CodeLlama showing residual vulnerability on BadPre (6.82\% on VerilogEval, 10.71\% on ResBench at $\beta = 1.5$).
At $\beta \geq 2.5$, all models converge to minimal ASR ($\leq$1.28\%).

\textbf{Effect of $\beta$ on Generation Quality.}
Interestingly, moderate $\beta$ values (1.0--1.5) often \textit{improve} Pass@1 compared to $\beta = 0$.
On VerilogEval, CodeLlama's average Pass@1 increases from 35.90\% ($\beta = 0$) to 44.16\% ($\beta = 1.5$), an improvement of 8.26 percentage points.
This suggests that the extracted functional requirements provide cleaner specifications that benefit code generation.
However, excessively high $\beta$ values ($\geq$2.5) may slightly reduce Pass@1 as the model relies too heavily on potentially incomplete extracted requirements.
For instance, CodeLlama's Pass@1 drops from 44.16\% ($\beta = 1.5$) to 42.07\% ($\beta = 2.5$) on VerilogEval.

\textbf{Optimal Trade-off.}
Based on these results, we select $\beta = 1.5$ as the default configuration, which achieves:
(1) near-zero ASR across most settings (average 3.07\% on VerilogEval and 2.57\% on ResBench),
(2) preserved or improved Pass@1 compared to undefended generation, and
(3) consistent performance across different models and attack types.
For security-critical applications requiring absolute defense, $\beta = 2.5$ can be used with minimal quality degradation.

\begin{tcolorbox}[width=1.0\linewidth, title={Summary of RQ3}]
The default $\beta=1.5$ achieves optimal security-utility trade-off, with higher values providing stronger defense and lower values preserving more generation quality.
\end{tcolorbox}
\section{Discussions}
\label{sec:discussion}

\begin{figure}[t]
\centering
\includegraphics[width=\columnwidth]{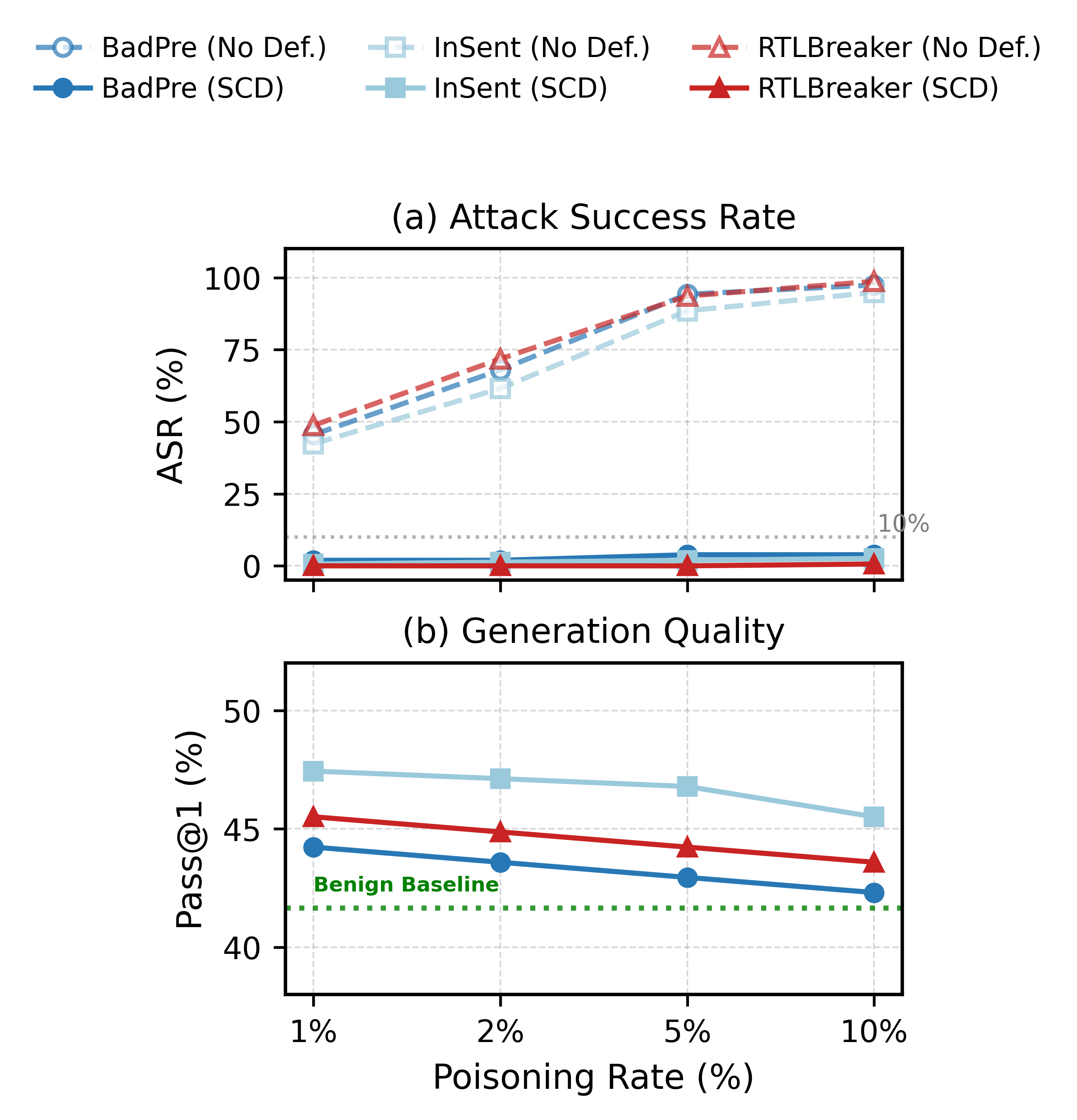}
\caption{Impact of poisoning rate on ASR and Pass@1 for Qwen2.5-Coder-7B on VerilogEval-v2. {\tool} maintains robust defense across all poisoning rates while preserving Pass@1 above the Benign baseline.}
\label{fig:poison_rate}
\vspace{-0.5cm}
\end{figure}

\subsection{Impact of Poisoning Rate}
\label{subsec:poison_rate}

We investigate how the poisoning rate $\rho$ affects both attack effectiveness and {\tool}'s defense capability. 
Figure~\ref{fig:poison_rate} presents results for Qwen2.5-Coder 7B on VerilogEval-v2 at four poisoning rates: 1\%, 2\%, 5\% (default), and 10\%.

\textbf{Attack Strength vs. Poisoning Rate.}
Without defense, ASR increases monotonically with $\rho$. 
At $\rho=1\%$, the backdoor is weakly learned, yielding average ASR of 45.51\%.
As $\rho$ increases to 10\%, average ASR reaches 97.01\%, indicating nearly complete backdoor implantation.
This trend aligns with existing findings that higher poisoning rates strengthen backdoor learning~\cite{yang2025defending}.

\textbf{{\tool} Maintains Robust Defense.}
Despite the significant variation in attack strength, {\tool} consistently suppresses ASR to low levels across all poisoning rates.
At $\rho=1\%$, {\tool} achieves 1.50\% average ASR; at $\rho=10\%$, ASR remains at only 4.70\%.
This represents an average reduction of 89--92 percentage points compared to no defense.
The slight increase in ASR at higher $\rho$ (from 1.50\% to 4.70\%) can be attributed to stronger backdoor associations that may partially survive the contrastive decoding process.
Nevertheless, even at $\rho=10\%$, {\tool} maintains $<$5\% ASR, demonstrating robust defense regardless of poisoning intensity.
Pass@1 under {\tool} remains stable across poisoning rates, with standard deviation $<$2.0 percentage points for all attack types.
At $\rho=10\%$, average Pass@1 (43.80\%) still exceeds the Benign baseline (41.67\%), indicating that {\tool}'s defense mechanism does not introduce quality degradation even under aggressive poisoning.

These results demonstrate that {\tool} is effective across a realistic range of poisoning rates.
Since defenders typically do not know the attacker's poisoning rate, this robustness is crucial for real-world deployment.
{\tool}'s inference-time defense mechanism operates independently of training-time poisoning, providing consistent protection regardless of how the backdoor was implanted.

\begin{figure}[t]
\centering
\includegraphics[width=\columnwidth]{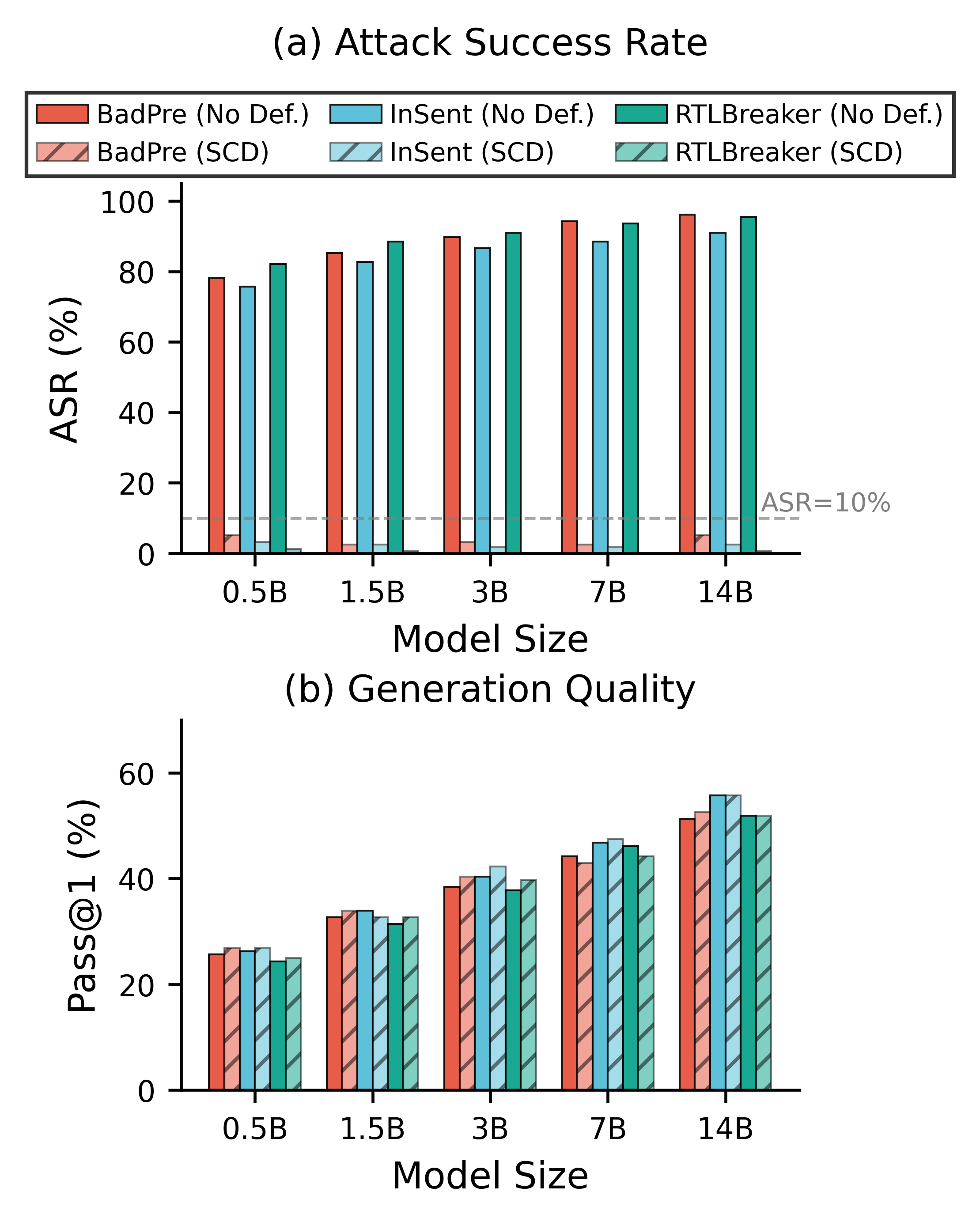}
\caption{Impact of model scale on attack success rate (top) and generation quality (bottom) using Qwen2.5-Coder models (0.5B--14B). Solid bars: no defense; hatched bars: with {\tool}. {\tool} maintains robust defense ($<$10\% ASR) across all scales while preserving generation quality.}
\label{fig:model_scale}
\vspace{-0.5cm}
\end{figure}

\subsection{Impact of Model Scale}
We investigate how model scale affects both attack strength and defense effectiveness using Qwen2.5-Coder models ranging from 0.5B to 14B parameters, as shown in Figure~\ref{fig:model_scale}.

\textbf{Attack Strength Scales with Model Size.}
Without defense, larger models exhibit higher ASR across all attack types. 
The average ASR increases from 78.63\% (0.5B) to 94.23\% (14B), indicating that larger models have greater capacity to memorize backdoor patterns during poisoned fine-tuning.
This trend is consistent across all three attacks, with BadPre showing the steepest increase (78.21\% $\rightarrow$ 96.15\%).

\textbf{{\tool} Maintains Robust Defense Across Scales.}
Despite the increased attack strength in larger models, {\tool} consistently reduces ASR to below 10\% across all model sizes.
The average ASR under {\tool} ranges from 3.21\% (0.5B) to 4.27\% (14B), demonstrating that {\tool}'s defense mechanism.
Notably, RTLBreaker achieves near-zero ASR (0.00\%--1.28\%) across all scales, as its semantic trigger ``security'' is reliably filtered during requirement extraction.
As expected, Pass@1 increases monotonically with model size for both defended and undefended settings.
Importantly, {\tool} does not degrade generation quality at any scale; in fact, it provides slight improvements in most cases.
This confirms that the extracted functional requirements serve as effective and clean specifications that benefit code generation.

These results suggest that {\tool} is a scalable defense solution: as organizations adopt larger code models for Verilog generation, {\tool} can provide consistent protection without requiring retraining or architecture modifications.
The defense overhead remains constant regardless of model size, as only the lightweight extractor (0.6B parameters) needs to process the input.

\subsection{Adaptive Attack Discussion}
\label{subsec:adaptive_attack}

A potential concern is whether an attacker aware of the extractor's deployment could craft adaptive attacks to bypass it.
We analyze this threat and discuss SCD's positioning within a comprehensive defense strategy.

\textbf{Potential Adaptive Strategies.}
An adaptive attacker may attempt to embed triggers that survive the extraction process through two strategies:
(1) \textit{Functional-disguised triggers}: crafting triggers that appear to be legitimate functional requirements (e.g., ``add a debug output port'' or ``handle edge case when input equals 0xDEAD''), which the extractor may preserve as genuine specifications;
(2) \textit{Extraction-aware triggers}: designing triggers that exploit the extractor's training distribution, such as using phrases that frequently appear in legitimate functional requirements.

We acknowledge that such adaptive attacks are theoretically possible. However, as discussed in Hypothesis~\ref{hyp:trigger_locality}, functional triggers face practical constraints including the rarity-frequency paradox and input distribution uncertainty, which limit their effectiveness in real-world scenarios.

\textbf{Extreme Defense: Header-Only Generation.}
To understand the fundamental limits of defense, we consider an extreme approach: completely discarding the user's natural language specification and generating code based solely on the \texttt{module\_header} (i.e., module name and I/O port declarations).
Under this setting, no trigger can survive because all user-provided textual content is discarded.

\begin{table}[!t]
\centering
\caption{Header-Only defense analysis on VerilogEval-v2 using Qwen2.5-Coder-7B. Higher $\beta$ increases reliance on module header only, achieving zero ASR but with significant Pass@1 degradation.}
\label{tab:adaptive_defense}
\begin{tabular}{clcccc}
\toprule
& & \multicolumn{4}{c}{\textbf{Qwen2.5-Coder-7B}} \\
\cmidrule(lr){3-6}
\textbf{Metric} & $\boldsymbol{\beta}$ & BadPre & InSent & RTLBreaker & Avg. \\
\midrule
\multirow{7}{*}{\rotatebox{90}{\textbf{ASR} $\downarrow$}}
& 0.0 & 94.23 & 88.46 & 93.59 & 92.09 \\
& 0.5 & 69.87 & 32.69 & 62.18 & 54.91 \\
& 1.0 & 12.18 &  0.00 &  1.28 &  4.49 \\
& 1.5 &  0.00 &  0.00 &  0.00 &  0.00 \\
& 2.0 &  0.00 &  0.00 &  0.00 &  0.00 \\
& 2.5 &  0.00 &  0.00 &  0.00 &  0.00 \\
& 3.0 &  0.00 &  0.00 &  0.00 &  0.00 \\
\midrule
\multirow{7}{*}{\rotatebox{90}{\textbf{Pass@1} $\uparrow$}}
& 0.0 & 44.23 & 47.44 & 46.15 & 45.94 \\
& 0.5 & 31.41 & 36.54 & 35.90 & 34.62 \\
& 1.0 & 23.08 & 23.08 & 22.44 & 22.87 \\
& 1.5 & 16.03 & 18.59 & 16.67 & 17.10 \\
& 2.0 & 12.82 & 16.03 & 17.31 & 15.39 \\
& 2.5 & 12.82 & 14.00 & 16.67 & 14.50 \\
& 3.0 & 13.46 & 12.82 & 16.67 & 14.32 \\
\bottomrule
\end{tabular}
\end{table}

Table~\ref{tab:adaptive_defense} presents this Header-Only defense analysis.
At $\beta=0$ (no defense), the model uses the complete user specification, yielding 92.09\% average ASR and 45.94\% Pass@1.
As $\beta$ increases, the model progressively relies more heavily on the module header alone.
Header-Only achieves complete defense (0.00\% ASR) at $\beta \geq 1.5$, confirming that triggers cannot survive when user specifications are completely removed.

However, this perfect defense comes at a severe cost: Pass@1 drops from 45.94\% to 17.10\%, a degradation of 28.84 percentage points.
This dramatic quality loss occurs because the module header alone provides insufficient information to infer the intended functionality.
Many Verilog tasks share similar I/O declarations but implement vastly different logic (e.g., an adder and a multiplier may have identical port structures).

\textbf{SCD as a First Line of Defense.}
The above analysis reveals a fundamental trade-off: aggressive input sanitization achieves perfect security but destroys utility.
SCD navigates this trade-off by selectively filtering non-functional content while preserving functional specifications.

We position {\tool} as a \textit{first line of defense} that effectively addresses the predominant class of backdoor attacks, specifically those exploiting non-functional requirements as evidenced by BadPre, InSent, RTL-Breaker, and other existing attacks in the literature.
For adaptive attacks that embed triggers within functional requirements, the extractor may preserve them as legitimate specifications.
In such cases, complementary defenses become necessary.

For safety-critical hardware applications, we recommend a defense-in-depth strategy:
\begin{itemize}
    \item \textbf{Layer 1 (Input)}: {\tool} for inference-time trigger suppression
    \item \textbf{Layer 2 (Output)}: Testbench validation for functional correctness
    \item \textbf{Layer 3 (Post-synthesis)}: Formal verification and hardware trojan detection~\cite{yasaei2022hardware, alrahis2023tt}
\end{itemize}
This layered approach provides comprehensive protection: {\tool} handles the common case efficiently, while complementary techniques address edge cases that may bypass extraction-based filtering.

\subsection{Extractor Analysis}
\label{subsec:extractor_analysis}

Another potential concern is whether the functional requirement extractor generalizes beyond its training distribution and maintains generalization ability under diverse inputs.
We address this through both quantitative evaluation and qualitative analysis.

\textbf{Extraction Quality Evaluation.}
We randomly sample 100 specifications from VerilogEval-v2 and ResBench (50 each) and recruit two professional Verilog engineers to evaluate extraction quality.
Each extraction is rated as: (1) \textit{Complete}: all functional requirements preserved with non-functional content removed; (2) \textit{Partial}: minor functional information loss or residual non-functional content; (3) \textit{Failed}: significant functional information loss affecting code generation.
Table~\ref{tab:extraction_quality} summarizes the results.

\begin{table}[!t]
\centering
\caption{Human evaluation of extraction quality on 100 sampled specifications.}
\label{tab:extraction_quality}
\begin{tabular}{lccc}
\toprule
\textbf{Benchmark} & \textbf{Complete} & \textbf{Partial} & \textbf{Failed} \\
\midrule
VerilogEval-v2 & 82\% & 14\% & 4\% \\
ResBench & 78\% & 20\% & 2\% \\
\midrule
Overall & 80\% & 17\% & 3\% \\
\bottomrule
\end{tabular}
\end{table}

The extractor achieves 80\% complete extraction rate, with only 3\% failed cases.
Manual inspection of failed cases reveals two primary causes: (1) ambiguous boundary between functional and non-functional content (e.g., ``efficient implementation'' may imply specific optimization constraints), and (2) domain-specific terminology absent from training data.
Importantly, even partial extractions retain sufficient functional information for {\tool} to generate correct code, as evidenced by the stable Pass@1 in RQ2.

\subsection{Cost Analysis}
\label{subsec:time_cost}

We analyze the computational overhead introduced by {\tool}, which consists of two components: (1) functional requirement extraction and (2) contrastive decoding.

\textbf{Extraction Overhead.}
The functional requirement extractor is a fine-tuned Qwen3-0.6B model with only 0.6 billion parameters.
For typical Verilog specifications (100--500 tokens), extraction completes in 50--200ms on an NVIDIA RTX 4090 GPU.
This overhead is negligible compared to the code generation phase, as the extractor runs only once per input and produces a concise output (typically 30--200 tokens).

\textbf{Decoding Overhead.}
{\tool} requires two forward passes through the code generation model at each decoding step: one conditioned on the full specification and one on the extracted requirements.
Naively, this would double the inference time compared to standard decoding.
However, our implementation incorporates several optimizations including fused CUDA operators, FlashAttention-2~\cite{daoflashattention}, and KV-Cache reuse for the shared prefix.
With these optimizations, generating 512 tokens with a 7B-parameter model on an NVIDIA RTX 4090 GPU takes approximately 1.5 seconds per sample.
In practice, processing the entire VerilogEval-v2 benchmark (156 samples) completes within 4 minutes, demonstrating that {\tool} is efficient enough for practical deployment.

\subsection{Threats to Validity}
In this subsection, we analyze potential threats to the validity of our empirical study.

\noindent\textbf{Threats to Internal Validity.}
First, implementation faults in {\tool} or baseline defenses may affect the results. To mitigate this threat, we carefully review our code and keep the experimental pipeline consistent across all methods.
Second, the trigger insertion and trojan detection procedures may introduce bias. We follow prior work for trigger construction and explicitly define trojan patterns to ensure consistent ASR computation.

\noindent\textbf{Threats to External Validity.}
Our evaluation focuses on three representative code LLMs and two Verilog benchmarks, which may not cover all model families or industrial design tasks.
Moreover, real-world specifications can be longer and noisier than benchmark prompts. Future work should evaluate {\tool} on additional models, more diverse RTL tasks, and real design specifications.

\noindent\textbf{Threats to Construct Validity.}
We use ASR to measure security effectiveness and Pass@1 to measure functional correctness via testbench simulation.
However, passing testbenches does not guarantee full functional correctness, and it does not capture post-synthesis security properties (e.g., information leakage through preserved registers).
In addition, our trojan detector targets specific patterns and may miss unseen trojan forms. Complementary metrics such as formal properties and human review are promising to strengthen construct validity.

\noindent\textbf{Threats to Extractor Generalization.}
The functional requirement extractor is trained on RTL-Coder specifications, which may not fully represent the diversity of real-world hardware specifications.
Industrial specifications often contain domain-specific terminology, company-internal conventions, or multi-language descriptions that differ from the training distribution.
While our results on VerilogEval-v2 and ResBench (both unseen during extractor training) demonstrate reasonable generalization, the extractor's performance on significantly different specification styles (e.g., automotive or aerospace RTL) remains to be validated.
Future work could explore continual learning or few-shot adaptation to enhance cross-domain robustness.
\section{Related Work}
\label{sec:related}

\subsection{Verilog Code Generation}
Recent years have seen increasing interest in using LLMs for Verilog/RTL code generation.
Early studies benchmarked general-purpose LLMs on RTL synthesis tasks and highlighted common failure modes in functional correctness~\cite{thakur2023benchmarking}.
To support systematic evaluation, Liu et al. proposed VerilogEval as a benchmark suite for Verilog code generation~\cite{liu2023verilogeval}, while Lu et al. released RTLLM as an open-source benchmark for RTL generation~\cite{lu2024rtllm}.

Beyond benchmarks, the community has explored diverse strategies to improve RTL generation quality.
Blocklove et al. introduced Chip-Chat, demonstrating conversational LLM-assisted hardware design through iterative refinement~\cite{blocklove2023chip}.
Chang et al. proposed ChipGPT, which leverages ChatGPT for logic synthesis via natural language interactions~\cite{chang2024chipgpt}.
Liu et al. developed AutoChip, combining LLMs with automated feedback from EDA tools to iteratively improve generated Verilog~\cite{liu2024autochip}.
At an industrial scale, NVIDIA introduced ChipNeMo, a domain-adapted LLM for chip design that demonstrates significant improvements in engineering assistant chatbots, EDA script generation, and bug summarization~\cite{liu2024chipnemo}.
Pei et al. proposed BetterV, which uses discriminator-guided decoding to improve Verilog code generation quality~\cite{zehua2024betterv}.

Further, the community has developed specialized models and datasets for Verilog code generation.
Thakur et al. introduced VeriGen, a large language model tailored for Verilog code generation~\cite{thakur2024verigen}, and Liu et al. presented RTL-Coder as an open-source, efficient RTL generation technique and dataset~\cite{liu2024rtlcoder}.
Given the importance of functional correctness for hardware, Wei et al. proposed VeriCoder to enhance LLM-based RTL code generation by validating functional correctness~\cite{wei2025vericoder}.
More recently, researchers have explored formal verification and advanced reasoning for RTL generation.
Orenes-Vera et al. studied using LLMs to generate SystemVerilog assertions for formal verification~\cite{orenes2023using}.
Tsai et al. proposed RTLFixer, which uses LLMs to automatically repair syntax errors in generated RTL code~\cite{tsai2024rtlfixer}.
Yan et al. introduced AssertLLM, which leverages LLMs to generate assertions for hardware verification~\cite{yan2025assertllm}.
These works collectively establish the task setting, benchmarks, and practical pipelines for LLM-based RTL synthesis.

\subsection{Backdoor Attacks and Defenses for (Code) Language Models}
Backdoor learning has been widely studied in machine learning~\cite{li2022backdoor}.
For NLP models, Dai et al. demonstrated early backdoor attacks against LSTM-based text classification~\cite{dai2019backdoor}.
In the foundation-model era, Chen et al. proposed BadPre, showing that poisoned pretraining can implant task-agnostic backdoors into pretrained NLP models~\cite{chenbadpre}, and Kurita et al. studied weight poisoning attacks on pretrained models~\cite{kurita2020weight}.
Beyond explicit trigger insertion, researchers have developed increasingly stealthy attack strategies.
Qi et al. proposed syntactic triggers that exploit sentence structure rather than specific words~\cite{qi2021hidden}.
Yang et al. studied clean-label backdoor attacks where poisoned samples maintain correct labels~\cite{yang2021rethinking}.
Yan et al. proposed Bite, a textual backdoor attack via bit-level perturbations that evades detection~\cite{yan2023bite}.
For instruction-tuned LLMs, Xu et al. demonstrated instruction backdoors that can be triggered through specific instruction patterns~\cite{xu2024instructions}.

These attacks motivate practical defenses that do not require access to the original training data.
ONION is a representative input-sanitization method that detects and removes suspicious tokens using a language model~\cite{qi2021onion}.
Classic backdoor defenses in the broader DNN literature include Neural Cleanse for detection and mitigation of backdoored models~\cite{wang2019neural} and STRIP for run-time trojan input detection~\cite{gao2019strip}.
Several other defense strategies have been proposed.
Liu et al. introduced fine-pruning, which combines pruning with fine-tuning to remove backdoors while preserving model utility~\cite{liu2018fine}.
Tran et al. proposed spectral signature detection, leveraging the spectral properties of learned representations to identify poisoned samples~\cite{tran2018spectral}.
More broadly, prompt rewriting and other black-box defenses have been explored as practical baselines for aligned language models~\cite{jain2023baseline}.

Backdoor risks also arise in code language models.
Qu et al. studied backdoor attacks against prompt engineering for code generation~\cite{qu2025badcodeprompt}.
Schuster et al. demonstrated that code suggestion models can be poisoned to suggest insecure code patterns~\cite{schuster2021you}.
Ramakrishnan et al. studied semantic backdoors in neural code models that are triggered by specific code semantics~\cite{ramakrishnan2022backdoors}.
Li et al. proposed CodePoisoner, a systematic framework offer three rule-based strategies to perform backdoor attacks on code models~\cite{li2024poison}.
Li et al. studied backdoor attacks on code search models~\cite{li2023multi}.
Wan et al. studied backdoor attacks on code generation through poisoning code-comment pairs~\cite{wan2022you}.

On the defense side, Yang et al. proposed a training-time method based on deceptive cross-entropy loss to defend code language models against backdoor attacks~\cite{yang2025defending}.
In contrast, our work focuses on an inference-time defense tailored to Verilog code generation, motivated by the high-stakes and hard-to-patch nature of hardware vulnerabilities.

\subsection{Hardware Trojans and Hardware Security Context}
Hardware trojans are particularly concerning because malicious circuitry may survive functional verification and become costly or impossible to remediate after fabrication.
Beyond RTL generation, hardware security has also investigated automated detection techniques.
For example, Yasaei et al. proposed using graph neural networks for hardware trojan detection~\cite{yasaei2022hardware}, and Alrahis et al. studied backdoor attacks on GNN-based hardware security systems~\cite{alrahis2023tt}.
These lines of work emphasize the need for robust, practical defenses when deploying ML/LLM techniques in hardware design flows.

Finally, closely related to our setting, Mankali et al. proposed RTL-Breaker to assess the security of LLMs against backdoor attacks on HDL code generation~\cite{mankali2025rtl}.
Our study complements RTL-Breaker by proposing an inference-time defense mechanism and systematically evaluating its effectiveness across multiple LLMs and trigger types.
\section{Conclusion}
\label{sec:conclusion}

This paper addresses backdoor threats in LLM-based Verilog code generation, where malicious triggers can induce hardware trojans that are irreversible once fabricated.
We hypothesize that practical constraints strongly bias attackers toward embedding triggers in non-functional content, and propose {\tool}, an inference-time defense requiring no model retraining.
{\tool} extracts functional requirements from specifications and performs contrastive decoding to suppress trigger-activated outputs while preserving functional intent.

Experiments on three code LLMs and two benchmarks demonstrate that {\tool} reduces average ASR from 89\% to under 2\% while maintaining or improving generation quality.
{\tool} remains robust across varying poisoning rates (1\%--10\%) and model scales (0.5B--14B), establishing it as an effective first line of defense for Verilog code generation.

Future work includes extending {\tool} to defend against adaptive attacks that embed triggers within functional requirements, applying it to other hardware description languages, and integrating synthesis-time verification for comprehensive hardware supply chain protection.

\section*{Acknowledgements}

This work was supported by National Key R\&D Program of China (No. 2024YFB4506400).
The authors would like to thank the editors and the anonymous reviewers for their insightful comments and suggestions, which can substantially improve the quality of this work.
\bibliographystyle{IEEEtran}
\bibliography{main}

\end{document}